\documentclass[11pt,letterpaper]{article}
\usepackage[margin=1in]{geometry}
\usepackage{amssymb}
\usepackage{graphicx}
\usepackage{algorithm}
\usepackage[noend]{algorithmic}
\usepackage{microtype}
\usepackage{subfig}
\usepackage{amsmath}
\usepackage{amsthm}
\usepackage{enumerate}
\usepackage{wrapfig}
\usepackage{multirow}
\usepackage{nicefrac}
\usepackage[colorinlistoftodos,prependcaption,textsize=tiny]{todonotes}
\usepackage[colorlinks,linkcolor=blue,filecolor=blue,citecolor=blue,urlcolor=blue,pdfstartview=FitH,pagebackref]{hyperref}
\usepackage[nameinlink]{cleveref}
\usepackage{authblk}
\usepackage{comment}
\usepackage{setspace}

\usepackage{array}
\newcolumntype{P}[1]{>{\raggedright\arraybackslash}p{#1}}

\usepackage{url}

\definecolor{forestgreen}{rgb}{0.13, 0.55, 0.13}
\definecolor{darkP}{rgb}{1, 0, 1}

\newtheorem{theorem}{Theorem}
\newtheorem{lemma}[theorem]{Lemma}
\newtheorem{claim}[theorem]{Claim}

\newtheorem{problem}[theorem]{Problem}

\newtheorem{remark}[theorem]{Remark}
\newtheorem{corollary}[theorem]{Corollary}

\graphicspath{{./figures/}}

\def\C{\mathcal{C}}
\def\G{\mathcal{G}}
\def\T{\mathcal{T}}
\def\P{\mathcal{P}}
\def\I{\mathcal{I}}
\def\D{\mathcal{D}}

\def\A{\mathcal{A}}
\def\R{{\mathbb R}}
\def\integers{{\mathbb Z}}

\def\eps{{\varepsilon}}

\def\polylog{\rm polylog}
\def\poly{\rm poly}

\newcommand{\frechet}{Fr\'echet}
\newcommand{\dfd}{d_{dF}}
\newcommand{\dtw}{d_{DTW}}
\newcommand{\dpt}{d_{p,2}}

\def\dfn#1{\emph{#1}}

\begin{document}

\title{Approximate Nearest Neighbor for Curves:\\ Simple, Efficient, and Deterministic }

\author[1]{Arnold Filtser}
\author[2]{Omrit Filtser}
\author[3]{Matthew J. Katz}

\affil[1]{Bar-Ilan University\\ \texttt{arnold273@gmail.com}}
\affil[2]{The Open University of Israel\\ \texttt{omrit.filtser@gmail.com}}
\affil[3]{Ben-Gurion University of the Negev\\\ \texttt{matya@cs.bgu.ac.il}}

\date{}

\maketitle
\setcounter{page}{0}
\thispagestyle{empty}

\begin{abstract}
	
	In the \emph{$(1+\eps,r)$-approximate near-neighbor} problem for curves (ANNC) under some similarity measure $\delta$, the goal is to construct a data structure for a given set $\C$ of curves that supports approximate near-neighbor queries: Given a query curve $Q$, if there exists a curve $C\in\C$ such that $\delta(Q,C)\le r$, then return a curve $C'\in\C$ with $\delta(Q,C')\le(1+\eps)r$.
	There exists an efficient reduction from the \emph{$(1+\eps)$-approximate nearest-neighbor} problem to ANNC, where in the former problem the answer to a query is a curve $C\in\C$ with $\delta(Q,C)\le(1+\eps)\cdot\delta(Q,C^*)$, where $C^*$ is the curve of $\C$ most similar to $Q$.
	
	Given a set $\C$ of $n$ curves, each consisting of $m$ points in $d$ dimensions, we construct a data structure for ANNC that uses $n\cdot O(\frac{1}{\eps})^{md}$ storage space and has $O(md)$ query time (for a query curve of length $m$), where the similarity measure between two curves is their discrete \frechet\ or dynamic time warping distance. Our method is simple to implement, deterministic, and results in an exponential improvement in both query time and storage space compared to all previous bounds.
	
	Further, we also consider the \emph{asymmetric} version of ANNC, where the length of the query curves is $k \ll m$, and obtain essentially the same storage and query bounds as above, except that $m$ is replaced by $k$.
	Finally, we apply our method to a version of approximate range counting for curves and achieve similar bounds.
\end{abstract}

\vfill
\begin{spacing}{0.5}
{\small \setcounter{tocdepth}{1} \tableofcontents}
\end{spacing}
\newpage

\section{Introduction}

Nearest neighbor search is a fundamental and well-studied problem that has various applications in machine learning, data analysis, and classification. This important task also arises in applications where the recorded instances are trajectories or polygonal curves modeling, for example, epigenetic and surgical processes, market value fluctuations, population growth, the number of the requests per hour received at some web-page, and even the response of a football player in a given situation.

Let $\C$ be a set of $n$ curves, each consisting of at most $m$ points in $d$ dimensions, and let $\delta$ be some distance measure for curves. In the \emph{nearest-neighbor} problem for curves, the goal is to construct a data structure for $\C$ that supports nearest-neighbor queries, that is, given a query curve $Q$ of length at most $m$, return the curve $C^* \in \C$ closest to $Q$ (according to $\delta$).
The approximation version of this problem is the \emph{$(1+\eps)$-approximate nearest-neighbor} problem, where the answer to a query $Q$ is a curve $C\in\C$ with $\delta(Q,C)\le(1+\eps)\delta(Q,C^*)$.
We study a decision version of this approximation problem, which is called the \emph{$(1+\eps,r)$-approximate near-neighbor} problem for curves (ANNC). Here, if there exists a curve in $\C$ that lies within distance $r$ of the query curve $Q$, one has to return a curve in $\C$ that lies within distance $(1+\eps)r$ of $Q$.
Note that there exists a reduction from the $(1+\eps)$-approximate nearest-neighbor problem to the $(1+\eps,r)$-approximate near-neighbor problem~\cite{HIM12,Indyk00,shakhnarovich2006nearest}, at the cost of an additional logarithmic factor in the query time and an $O(\log^2 n)$ factor in the storage space.

In practice, it is often the case that the query curves are significantly shorter than the input curves (e.g., Google-search queries). Thus, we also study the \emph{asymmetric setting} of $(1+\eps,r)$-ANNC, where each of the input curves has complexity at most $m$, while each query curve has complexity at most $k\ll m$.

There are many methods that are used in real-world applications for comparing curves, and one of the most prevalent is the (discrete) \emph{\frechet} distance (DFD for short), which is often described by the following analogy. Two frogs are hopping from vertex to vertex along two polygonal curves. At each step, one of the frogs or both frogs may advance to the next vertex on its curve. The discrete \frechet\ distance is defined as the smallest maximum distance between the frogs that can be achieved in such a joint sequence of hops.
Another useful distance measure for curves or time series is the \emph{dynamic time warping} distance (DTW for short), in which instead of taking the smallest maximum distance we take the smallest sum of distances.

In the last several years, a series of papers have been written investigating the approximate near-neighbor problem for curves (ANNC) and its variants under the \frechet\ distance \cite{AD18,AFHKS19,dBCG13,dBGM17,DS17,EP20,Indyk02} (see~\Cref{tbl:results}), and several different approaches and sophisticated methods were utilized in order to provide efficient data structures. Up to now, all data structures for ANNC under DFD have either an exponential in $m$ query time, or an infeasible storage space bound. In this paper, for the first time, we manage to remove the exponential factor from the query time, while also significantly reducing the space consumption. Our approach consists of a discretization of space based on the input curves, which allows us to prepare a small set of curves that captures all possible queries approximately.

Indyk~\cite{Indyk02} was the first to give a deterministic near-neighbor data structure for curves under DFD. The data structure achieves an approximation factor of $O((\log m + \log \log n)^{t-1})$ given some trade-off parameter $t>1$. Its space consumption is very high, $O(m^2|X|)^{tm^{1/t}}\cdot n^{2t}$, where $|X|$ is the size of the domain on which the curves are defined, and the query time is $(m\log n)^{O(t)}$.
In \Cref{tbl:results} we set $t = 1 + o(1)$ to obtain a constant approximation factor.

Later, Driemel and Silvestri~\cite{DS17} presented a locality-sensitive-hashing scheme for curves under DFD, improving the result of Indyk for short curves. Their data structure uses $O(2^{4md}n\log n)$ space and answers queries in $O(2^{4md}\log n)$ time with an approximation factor of $O(d^{3/2})$. They also provide a trade-off between approximation quality and computational performance: for $d=O(1)$, and given a parameter $k\in[m]$, a data structure of size $O(2^{2k}m^{k-1} n\log n + mn)$ is constructed that answers queries in $O(2^{2k}m^k \log n)$ time with an approximation factor of $O(m/k)$. For constant $k$, this data structure uses only $\poly(m)\cdot n\log n$ space and has $\poly(m)\cdot\log n$ query time, however, the approximation factor is $O(m)$. They also show that this result can be applied to DTW, but only for the extreme of the trade-off which gives an $O(m)$ approximation.

Recently, Emiris and Psarros~\cite{EP20} presented near-neighbor data structures for curves under both DFD and DTW. Their algorithm provides an approximation factor of $(1+\eps)$, at the expense of increased space usage and preprocessing time. They use the idea that for a fixed alignment between two curves (i.e., a given sequence of hops of the two frogs), the problem can be reduced to the near-neighbor problem for points in $\ell_\infty$-product of $\ell_2$ spaces.
Their basic idea is to construct a data structure for every possible alignment. Once a query is given, they query all these data structures and return the closest curve found. This approach is responsible for the $2^m$ factor in their query time. 
Furthermore, they generalize this approach using randomized projections of $\ell_p$-products of Euclidean metrics (for any $p\ge1$), and define the $\ell_{p,2}$-distance for curves (for $p\ge 1$), which is exactly DFD when $p=\infty$, and DTW when $p=1$ (see \Cref{sec:prelim}). The space used by their data structure is $\tilde{O}(n)\cdot(2+\frac{d}{\log m})^{O(m^{1+1/\eps}\cdot d\log(1/\eps))}$ with query $\tilde{O}(dm^{1+1/\eps}\cdot 2^{4m}\log n)$ for DFD and $\tilde{O}(n)\cdot \frac{1}{\eps}^{O(md)}$ space and $\tilde{O}(d\cdot 2^{4m}\log n)$ query for DTW.

Har-Peled and Kumar~\cite{HK11} considered approximate nearest-neighbor for general metric spaces where the query points are constrained to lie on a subspace of low doubling dimension. It is easy to show that the doubling dimension of the metric space for curves of length $m$ in $d$ dimensions under DFD, is bounded by $O(md)$. Therefore, their result implies that one can construct a $(1+\eps)$-approximate nearest-neighbor data-structure requiring space in $n\eps^{-O(md)}$ and with query time in $2^{O(md)}\log(n/\eps)$. Note that in \cite{DS17} and \cite{EP20} the query time is exponential in $md$ and $m$, respectively, while the space complexity is suboptimal (and in \cite{DS17}, the approximation factor is $O(d^{3/2})$ and not $(1+\eps)$). In this paper, we get the exact same space as in ~\cite{HK11} (up to a $\log n$ factor), which seems to be optimal, while the query time is linear in $md$ (using the more structured input).

\paragraph{Subsequent work.}In a recent work, Bringmann et al.~\cite{BDNP22} study the asymmetric setting of $(1+\eps,r)$-ANNC under continuous \frechet\ distance, for one-dimensional polygonal curves (time series). Improving upon the work of Driemel and Psarros~\cite{DP21}, they use the idea of signatures to obtain some sort of a discrete version of the problem, and then follow our approach of preparing in advance the answers to all relevant queries on a discretization of the space (which was also used in~\cite{DP21}), to construct a data structure with space in $n\cdot O\left( \frac{m}{k\eps}\right)^{k}$ and query time in $O(1)^k$. They also show that an approximation factor of $(2+\eps)$ can be obtained with the same space bound and $O(k)$ query time. This provides additional evidence that our approach to ANNC, although quite simple and easy to implement, seems to produce more efficient data structures than those obtained using tools such as LSH and randomized projections. 
In addition, \cite{BDNP22} present conditional lower bounds for several different settings of the problem.

\paragraph{Our results.}
We present a data structure for the $(1+\eps,r)$-approximate near-neighbor problem using a bucketing method. We construct a relatively small set of curves $\I$, such that, given a query curve $Q$, if there exists some curve in $\C$ within distance $r$ of $Q$, then one of the curves in $\I$ must be very close to $Q$. The points of the curves in $\I$ are chosen from a simple discretization of space, thus, while it is not surprising that we get the best query time, it is surprising that we achieve a better space bound. Moreover, while the analysis of the space bounds is rather involved, the implementation of our data structures remain simple in practice.

See \Cref{tbl:results} for a summary of our results.
In the table, we do not state our result for the general $\ell_{p,2}$-distance. 
Instead, we state our results for the two most important cases, i.e. DFD and DTW, and compare them with previous work. 
Note that our results substantially improve the current state of the art for any $p\ge 1$.
In particular, we remove the exponential dependence on $m$ in the query bounds and significantly improve the space bounds.
 
Our results for the asymmetric setting, where the query curve $Q$ has complexity $k\ll m$, are summarized in \Cref{tbl:AssResults}. We show that in the asymmetric setting for DFD, our data structure can be slightly modified in order to achieve query time and storage space independent of $m$. Moreover, the storage space and query time matches those of the symmetric setting, by replacing $m$ with $k$.

We also apply our methods to an approximation version of range counting for curves (for the general $\ell_{p,2}$ distance) and achieve bounds similar to those of our ANNC data structure.
Moreover, at the cost of an additional $O(n)$-factor in the space bound, we can also answer the corresponding approximation version of range searching, thus answering a question of Afshani and Driemel~\cite{AD18}, with respect to DFD.

We note that our approach with obvious modifications works also in a dynamic setting, that is, we can construct an efficient dynamic data structure for ANNC as well as for other related problems such as range counting and range reporting for curves.

Another significant advantage of our approach is that, unlike some of the previous solutions, our data structure always returns an answer, and never returns a curve at distance greater than $(1+\eps)r$ from the query curve, i.e., there are no false positives. This is an important property of our solution, due to the fact that verifying the validity of the answer (i.e., computing the distance between two curves) cannot be done in strongly subquadratic time (assuming SETH, see~\cite{Bringmann14}), which is already more than our query time (for $d<m$).

\begin{table}[t]
	\begin{tabular}{ | l | l | l | l | P{3.35cm} |}
		\hline
		& Space & Query & Approx. & Comments \\ \hline\hline
		\multirow{4}{*}{DFD}  
		& $O(m^2|X|)^{m^{1-o(1)}}\cdot n^{2+o(1)}$ & $(m \log n)^{O(1)}$ & $O(1)$ &  {\scriptsize deterministic, \cite{Indyk02}} \\ \cline{2-5}
		& $O(2^{4md}n\log n)$ & $O(2^{4md}\log n)$ & $O(1)$ & \multirow{2}{3.35cm}{\scriptsize randomized, using LSH, $d=O(1)$, \cite{DS17}} \\ \cline{2-4}
		& $O(n\log n+mn)$ & $O(m\log n)$ & $O(m)$ & \\ \cline{2-5}
		& {$\tilde{O}(n)\cdot(2+\frac{d}{\log m})^{O(m^{O(\frac{1}{\eps})} d\log(\frac{1}{\eps}))}$} & {$\tilde{O}(dm^{O(\frac{1}{\eps})}\cdot 2^{4m}\log n)$} & $1+\eps$ & {\scriptsize randomized, \cite{EP20}} \\ \cline{2-5}
		& $n\cdot O(\frac{1}{\eps})^{md}$ & $O(md)$ & $1+\eps$ & {\scriptsize deterministic (randomized construction),
			\Cref{thm:DFD}}\\ \hline\hline
		
		\multirow{4}{*}{DTW}  & $O(n\log n+mn)$ & $O(m\log n)$ & $O(m)$ & {\scriptsize randomized, using LSH, $d=O(1)$, \cite{DS17}} \\ \cline{2-5}
		& $\tilde{O}(n)\cdot\frac{1}{\eps}^{O(md)}$ & $\tilde{O}(d\cdot2^{4m}\log n)$ & $1+\eps$ & {\scriptsize randomized, \cite{EP20}} \\ \cline{2-5}
		& $n\cdot O(\frac{1}{\eps})^{m(d+1)}$ & $O(md)$ & $1+\eps$ & {\scriptsize deterministic (randomized construction),
			\Cref{thm:DTW}}\\ \hline
	\end{tabular}
	\caption{\label{tbl:results}\small Our approximate near-neighbor data structure under DFD and DTW compared to the previous results.}
	\vspace{-10pt}
\end{table}

\begin{table}[t]\begin{center}		
		\begin{tabular}{ | l | l  | P{2.3cm} |l|}
			\hline Space & Query  & Deterministic construction?& Reference \\ \hline \hline
			$n\cdot \left(O(\frac{kd^{3/2}}{\eps})^{kd}\right)$ &  $O(kd)$ & no& \cite{DPS19}\\\hline
			$n\cdot \left(O(\frac{kd^{3/2}}{\eps})^{kd+1}\right)$ &  $O(\frac{k^2d^{5/2}}{\eps}(\log n+kd\log(\frac{kd}{\eps})))$ & yes& \cite{DPS19}\\\hline
			$n\cdot O(\frac{1}{\eps})^{kd}$ & $O(kd)$ & no& \Cref{thm:DFDassym}\\\hline
			$n\cdot O(\frac{1}{\eps})^{kd}$ & $O(kd\log(\frac{nkd}{\eps}))$ & yes& \Cref{thm:DFDassym-deterministic}\\ \hline
		\end{tabular}
	\end{center}
	\vspace{-11pt}
	\caption{\label{tbl:AssResults}\small Summary of previous and current results for the asymmetric approximate near-neighbor data structure for curves. All the results in the table are w.r.t. DFD. 
		Setting: the input consists of $n$ curves with $m$ points each, the query curve is of length $k$, 
		the approximation ratio is $1+\eps$ for $\eps\in (0,1)$, and our data structures always succeed. 
		Historic note: \cite{DPS19} is a subsequent work to the first version of this paper. In this version we also apply our counting techniques to the asymmetric cases.}
	\vspace{-10pt}
\end{table}

\subsection{Related work}\label{sec:related}
De Berg, Gudmundsson, and Mehrabi~\cite{dBGM17} described a dynamic data structure for approximate nearest neighbor for curves (which can also be used for other types of queries such as range reporting), under the (continuous) \frechet\ distance. 
Their data structure uses $n\cdot O\left(\frac{1}{\eps}\right)^{2m}$ space and has $O(m)$ query time, but with an additive error of $\eps \cdot reach(Q)$, where $reach(Q)$ is the maximum distance between the start vertex of the query curve $Q$ and any other vertex of $Q$. Furthermore, their query procedure might fail when the distance to the nearest neighbor is relatively large.

Afshani and Driemel~\cite{AD18} studied (exact) range searching under both the discrete and continuous \frechet\ distance. In this problem, the goal is to preprocess $\C$ such that given a query curve $Q$ of length $m_q$ and a radius $r$, all the curves in $\C$ that are within distance $r$ from $Q$ can be found efficiently.
For DFD, their data structure uses $O(n(\log \log n)^{m-1})$ space and has $O(n^{1-\frac{1}{d}}\cdot\log^{O(m)}n\cdot m_q^{O(d)})$ query time, where $m_q$ is limited to $\log^{O(1)}n$. 
Additionally, they provide a lower bound in the pointer model, stating that every data structure with $Q(n) + O(k)$ query time, where $k$ is the output size, has to use roughly $\Omega\left((n/Q(n))^2\right)$ space in the worst case (even for $m_q=1$).
Afshani and Driemel conclude their paper by asking whether more efficient data structures might be constructed if one allows approximation.

De Berg, Cook IV, and Gudmundsson~\cite{dBCG13}, considered the following range counting problem under the continuous \frechet\ distance. 
Given a polygonal curve $C$ with $m$ vertices, they show how to preprocess it into a data structure of size $O(k \cdot\polylog (m))$, so that, given a query segment $s$, one can return a constant approximation of the number of subcurves of $C$ that lie within distance $r$ of $s$ in $O(\frac{m}{\sqrt{k}} \cdot\polylog (m))$ time, where $k$ is a parameter between $m$ and $m^2$.

Aronov et al.~\cite{AFHKS19} managed to obtain practical bounds for two cases of the asymmetric $(1+\eps,r)$-ANNC under DFD: (i) when $Q$ is a line segment (i.e., $k=2$), or (ii) when $\C$ consists of line segments (i.e., $m=2$). The bounds on the size of the data structure and query time are nearly linear in the size of the input and query curve, respectively. Specifically, for the case where $k=2$, they achieve query time $O(\log^4(\frac{n}{\eps}))$ and storage space $O(n\frac{1}{\eps^4}\log^4(\frac{n}{\eps}))$.
They also  provide efficient data structures for several other variants of the problem: the (exact) NNC where $\ell_\infty$ is used for interpoint distances, and the case where the location of the input curves is only fixed up to translation.

\subsection{Technical ideas}
We use a discretization of the space, by laying a $d$-dimensional uniform grid with edge length $\frac{\eps r}{\sqrt{d}}$. The main ingredient in our data structure is then a relatively small set $\I$ of curves defined by grid points, which represents all possible queries. For each curve in $\I$ we store an index of a close enough curve from the input set $\C$. Given a query $Q$ sufficiently close to some curve in $\C$, we find a representative $Q'$ in $\I$ by simply rounding $Q$'s vertices and return the index of the curve stored for $Q'$.

Given a point $x\in \R^d$, the number of grid points that are within distance $(1+\eps)r$ from $x$ is bounded by $O(\frac{1}{\eps})^d$ (\Cref{cor:gridDFD}). Thus, given a curve $C$ of length $m$, the total number of grid points that are within distance $(1+\eps)r$ from one of its vertices is $m\cdot O(\frac{1}{\eps})^d$. Naively, the number of curves needed to represent all possible queries of length $m$ within distance $r$ of $C$ is bounded by the number of ways to choose $m$ points with repetitions from a set of grid points of size $m\cdot O(\frac{1}{\eps})^d$, which is bounded by $m^m\cdot O(\frac{1}{\eps})^{md}$.
This infeasible bound on the storage space might be the reason why more sophisticated solutions for ANNC have been suggested throughout the years.

One of the main technical contributions of this paper is an analysis leading to a significantly better bound, if we store only candidate curves that are within distance $(1+\eps)r$ from $C$. Actually, in~\Cref{sec:frechet} we show that for the case of DFD, it is sufficient to store a set of representative curves of size only $O(\frac{1}{\eps})^{md}$ for each input curve. The basic idea is to bound the number of representatives that can be obtained by some fixed alignment between $C$ and the candidate curve (see~\Cref{clm:curvesDFD}).

For the general case of $\ell_{p,2}$-distance (including DTW), we are minimizing the sum of distances instead of the maximum distance (as in DFD). Thus, we have to use a more dense grid (with edge length $\frac{\eps r}{(2m)^{1/p}\sqrt{d}}$), and the situation becomes more complicated. First, unlike DFD, the triangle inequality does not hold for $\ell_{p,2}$-distance in general (including DTW). Second, since DFD is a min-max measure, the choice of different vertices for a representative curve is ``independent'' in a sense, whereas for $\ell_{p,2}$-distance in general, the choice of different vertices depends on their sum of distances from the input curve. Using more careful counting arguments and analysis of the alignment between two curves, we are able to show that in this case the number of representative curves that our data structure has to store per input curve is bounded by $O(\frac{1}{\eps})^{m(d+1)}$ (see~\Cref{clm:curves}).

To store the set $\I$ we simply use a dictionary, which can be implemented using a hash table and guarantees a query time linear in the size of the query. To obtain a fully deterministic solution, one can use a search tree instead. However, a naive implementation using a binary search tree results in an additional factor of $O(\log|\I|)=O(md\log(\frac{n}{\eps}))$ to the query time, i.e., in a query time of $O(m^2d^2\log(\frac{n}{\eps}))$. We show how to implement the dictionary using a prefix tree, exploiting the fact that the vertices of the curves in $\I$ are from a relatively small set of grid points, which improves the query time to $O(md\log(\frac{nmd}{\eps}))$.

For the asymmetric setting (where the length of a query is $k\ll m$), we use simplifications of the input curves in order to obtain bounds that are independent of $m$. Given a curve $C$ of length $m$, a simplification $\Pi$ of $C$ is a curve of length $k\ll m$ that is relatively close to $C$. Simplifications were used in order to provide approximate solutions in several asymmetric versions of problems on curves, such as clustering~\cite{BDGHKLS19}, and distance oracles~\cite{DH13,DPS19}.

By the triangle inequality for DFD, every query curve $Q$ within distance $r$ from an input curve $C$ is at distance at most $2r$ from the simplification $\Pi$ (where $\Pi$ is within distance $r$ from $C$). Thus, it is enough to prepare for query curves at distance at most $2r$ from $\Pi$, which follows from previous arguments. Note that the query time and storage space are independent of $m$. For the asymmetric setting under general $\ell_{p,2}$-distance, the situation again becomes more complicated. First, we present an algorithm that computes the closest (vertex-restricted) simplification of length $k$ to a given curve of length $m$ under the $\ell_{p,2}$-distance. 
	In order to adapt our data structure to the asymmetric setting, we need to increase the allowed distance between a simplification and a representative curve by a factor of $k^{1/p}$, for the triangle inequality to work. The counting arguments that we use for the symmetric case yield a bound of $O(\frac{k^{1/p}}{\eps})^{k(d+1)}$ on the storage space. In \Cref{subsec:asym_lp}, we provide stronger counting arguments, which enable us to remove the $k^{1/p}$ factor from the base of the exponent. The main idea is to use the simplification in order to divide the input curve into $O(k)$ compact subsequences (see \Cref{clm:curve-partition}).

\section{Preliminaries}\label{sec:prelim}
To simplify the presentation, we assume throughout the paper that all the input curves have exactly the same size, $m$, and all the query curves have exactly the same size, either $m$ or $k$, depending on whether we are considering the standard or the asymmetric version. This assumption can be easily removed (see \Cref{appendix:DifferentSizes}).

Let $\C$ be a set of $n$ curves, each consisting of $m$ points in $\mathbb{R}^d$, and let $\delta$ be some distance measure for curves. 

\begin{problem}[$(1+\eps)$-approximate nearest-neighbor for curves]
	Given a parameter $0<\eps\le 1$, preprocess $\C$ into a data structure that given a query curve $Q$, returns a curve $C' \in \C$, such that $\delta(Q,C') \le (1+\eps)\cdot\delta(Q,C)$, where $C$ is the curve in $\C$ closest to $Q$.
\end{problem}

\begin{problem}[$(1+\eps,r)$-approximate near-neighbor for curves]
	Given a parameter $r$ and $0<\eps\le 1$, preprocess $\C$ into a data structure that given a query curve $Q$, if there exists a curve $C_i\in \C$ such that $\delta(Q,C_i)\le r$, returns a curve $C_j\in \C$ such that $\delta(Q,C_j)\le (1+\eps)r$.
\end{problem}

\begin{problem}[Asymmetric $(1+\eps,r)$-approximate near-neighbor for curves]
	Given parameters $r$,$k$, and $0<\eps\le 1$, preprocess $\C$ into a data structure that given a query curve $Q$ of length $k$, if there exists a curve $C_i\in \C$ such that $\delta(Q,C_i)\le r$, returns a curve $C_j\in \C$ such that $\delta(Q,C_j)\le (1+\eps)r$.
\end{problem}

\paragraph{Curve alignment.}
Given two integers $m_1,m_2$, let $\tau := \langle(i_1,j_1),\dotsc,(i_t,j_t)\rangle$ be a sequence of pairs where $i_1=j_1=1$, $i_t=m_1$,$j_t = m_2$, and for each $1 < k \leq t$, one of the following conditions holds:
\begin{enumerate}[(i)]
	\item $i_k = i_{k-1} + 1$ and $j_k = j_{k-1}$,
	\item $i_k = i_{k - 1}$ and $j_k = j_{k-1} + 1$, or
	\item $i_k = i_{k-1} + 1$ and $j_k = j_{k-1}+1$.
\end{enumerate}
We call such a sequence $\tau$ an \emph{alignment} w.r.t. two curves of lengths $m_1$ and $m_2$, respectively.

\hfill \\
\noindent 
Let $P = (p_1,\dotsc,p_{m_1})$ and $Q = (q_1,\dotsc,q_{m_2})$ be two curves of lengths $m_1$ and $m_2$, respectively, in $\mathbb{R}^d$. We say that an alignment $\tau$ w.r.t. $P$ and $Q$ matches $p_i$ and $p_j$ if $(i,j)\in\tau$.

\paragraph{Discrete \frechet\ distance (DFD).}
The \dfn{\frechet\ cost} of an alignment $\tau$ w.r.t. $P$ and $Q$ is $\sigma_{dF}(\tau(P,Q)):=\max_{(i,j)\in\tau} \|p_i-q_j\|_2$.
The discrete \frechet\ distance is defined over the set $\T$ of all alignments as 
\[\dfd(P,Q)=\min_{\tau \in \T}\sigma_{dF}(\tau(P,Q)).\]

\vspace{-10pt}
\paragraph{Dynamic time wrapping (DTW).}
The \dfn{time warping cost} of an alignment $\tau$ w.r.t. $P$ and $Q$ is $\sigma_{DTW}(\tau(P,Q)):=\sum_{(i,j)\in\tau} \|p_i-q_j\|_2$.
The DTW distance is defined over the set $\T$ of all alignments as 
\[\dtw(P,Q)=\min_{\tau \in \T}\sigma_{DTW}(\tau(P,Q)).\]

\vspace{-10pt}
\paragraph{$\ell_{p,2}$-distance for curves.} 
The \dfn{$\ell_{p,2}$-cost} of an alignment $\tau$ w.r.t. $P$ and $Q$ is $\sigma_{p,2}(\tau(P,Q)):=\left(\sum_{(i,j)\in\tau}\|p_i-q_j\|_2^p\right)^{1/p}$.
The $\ell_{p,2}$-distance between $P$ and $Q$ is defined over the set $\T$ of all alignments as 
\[
\dpt(P,Q)=\min_{\tau \in \T}\sigma_{p,2}(\tau(P,Q)).
\]

Notice that $\ell_{p,2}$-distance is a generalization of DFD and DTW, in the sense that $\sigma_{dF}=\sigma_{\infty,2}$ and $\dfd=d_{\infty,2}$, $\sigma_{DTW}=\sigma_{1,2}$ and $\dtw=d_{1,2}$. Also note that DFD satisfies the triangle inequality, but DTW and $\ell_{p,2}$-distance (for $p\ne\infty$) do not (see~\Cref{sec:l_p2-distance} for details).

Emiris and Psarros~\cite{EP20} showed that the number of all possible alignments of two curves is in $O(m\cdot 2^{2m})$. We reduce this bound by counting only alignments that can determine the $\ell_{p,2}$-distance between two curves.\footnote{Since our storage space is already in $O(\frac{1}{\eps})^{md}$, and $m\cdot 2^{2m}\le 3^{2m}$ is in $O(1)^{md}$,
we could have used this larger upper bound. However, in \Cref{lem:alignments} we show a tight upper bound on the number of relevant alignments, which may be useful for other applications.} More formally, let $\tau$ be an alignment. If there exists an alignment $\tau'$ such that $\tau'\subset\tau$, then clearly $\sigma_{p,2}(\tau'(P,Q))\le\sigma_{p,2}(\tau(P,Q))$, for any $1\le p \le \infty$ and for any two curves $P$ and $Q$. In this case, we say that $\tau$ cannot determine the $\ell_{p,2}$-distance between two curves.
\begin{lemma}\label{lem:alignments}
	The number of different alignments that can determine the $\ell_{p,2}$-distance between two $m$-curves (for any $1\le p \le \infty$) is at most $O(\frac{2^{2m}}{\sqrt{m}})$.
\end{lemma}
\begin{proof}
	Let $\tau = \langle(i_1,j_1),\dotsc,(i_t,j_t)\rangle$ be an alignment. Notice that $m\le t\le 2m-1$.
	By definition, $\tau$ has 3 types of (consecutive) subsequences of length two:
	\begin{enumerate}[(i)]
		\item $\langle(i_k,j_k),(i_k+1,j_k)\rangle$,
		\item $\langle(i_k,j_k),(i_k,j_k+1)\rangle$, and
		\item $\langle(i_k,j_k),(i_k+1,j_k+1)\rangle$.
	\end{enumerate}
	
	Denote by $\T_1$ the set of all alignments that do not contain any subsequence of type (iii). Then, any $\tau_1\in \T_1$ is of length exactly $2m-1$. Moreover, $\tau_1$ contains exactly $2m-2$ subsequences of length two, of which $m-1$ are of type (i) and $m-1$ are of type (ii). Therefore, $|\T_1|={2m-2 \choose m-1}=O(\frac{2^{2m}}{\sqrt{m}})$.
	
	Assume that an alignment $\tau$ contains a subsequence of the form $(i_k,j_k-1),(i_k,j_k),(i_k+1,j_k)$, for some $1 < k \le t-1$. Notice that removing the pair $(i_k,j_k)$ from $\tau$ results in a legal alignment $\tau'$, such that $\sigma_{p,2}(\tau'(P,Q))\le\sigma_{p,2}(\tau(P,Q))$, for any $1\le p \le \infty$ and two curves $P,Q$. We call the pair $(i_k,j_k)$ a \emph{redundant pair}. Similarly, if $\tau$ contains a subsequence of the form $(i_k-1,j_k),(i_k,j_k),(i_k,j_k+1)$, for some $1 < k \le t-1$, then the pair $(i_k,j_k)$ is also a redundant pair. Therefore we only care about alignments that do not contain any redundant pairs.
	Denote by $\T_2$ the set of all alignments that do not contain a redundant pair, then any $\tau_2\in \T_2$ contains at least one subsequence of type (iii).
	
	We claim that for any alignment $\tau_2\in\T_2$, there exists a unique alignment $\tau_1\in\T_1$. Indeed, if we add the redundant pair $(i_l,j_l+1)$ between $(i_l,j_l)$ and $(i_l+1,j_l+1)$ for each subsequence of type (iii) in $\tau_2$, we obtain an alignment $\tau_1\in\T_1$. Moreover, since $\tau_2$ does not contain any redundant pairs, the reverse operation on $\tau_1$ results in $\tau_2$. Thus we obtain $|\T_2|\le |\T_1|=O(\frac{2^{2m}}{\sqrt{m}})$.
\end{proof}
\paragraph{Points and balls.}
Given a point $x\in\R^d$ and a real number $R > 0$, we denote by $B^d_p(x,R)$ the $d$-dimensional ball under the $\ell_p$ norm with center $x$ and radius $R$, i.e., a point $y\in \R^d$ is in $B^d_p(x,R)$ if and only if $\|x-y\|_p\le R$, where $\|x-y\|_p=\left(\sum_{i=1}^{d}|x_i-y_i|^p\right)^{1/p}$.
Let $B^d_p(R)=B^d_p(\textbf{0},R)$, and let $V^d_p(R)$ be the volume (w.r.t. Lebesgue measure) of $B^d_p(R)$, then 
\[V^d_p(R)=\frac{2^d\Gamma(1+1/p)^{d}}{\Gamma(1+d/p)}R^d,\] 
where $\Gamma(\cdot)$ is Euler's Gamma function (an extension of the factorial function).
For $p=2$ and $p=1$, we get \[V^d_2(R)=\frac{\pi^{d/2}}{\Gamma(1+d/2)}R^d \mbox{\ \ \ and\ \ \ } V^d_1(R)=\frac{2^d}{d!}R^d.\]

Our approach consists of a discretization of the space using lattice points, i.e., points from $\mathbb{Z}^d$.
\begin{lemma}\label{lem:lattice}
	The number of lattice points in the $d$-dimensional ball of radius $R$ under the $\ell_p$ norm (i.e., in $B^d_p(R)$)
	is bounded by $V^d_p(R+d^{1/p})$.
\end{lemma}
\begin{proof}
	With each lattice point $z=(z_1,z_2,\dots,z_d)$, $z_i\in\integers$, we match the $d$-dimensional lattice cube $C(z)=[z_1,z_1+1]\times[z_2,z_2+1]\times\dots\times[z_d,z_d+1]$. Notice that $z\in C(z)$, and the $\ell_p$-diameter of a lattice cube is $d^{1/p}$. Therefore, the number of lattice points in the $\ell^d_p$-ball of radius $R$ is bounded by the number of lattice cubes that are contained in a $\ell^d_p$-ball with radius $R+d^{1/p}$. This number is bounded by $V^d_p(R+d^{1/p})$ divided by the volume of a lattice cube, which is $1^d=1$.
\end{proof}

\begin{remark}
	In general, in all our data structures we
	do not assume any bound on the dimension $d$. However, using dimension reduction techniques, we may assume that $d\le O(\frac{\log (nm)}{\eps^2})$.
 See \Cref{appendix:DimensionReduction} for details.
\end{remark}

\section{Discrete \frechet\ distance (DFD)}\label{sec:frechet}
Consider the infinite $d$-dimensional grid with edge length $\frac{\eps r}{\sqrt{d}}$.
Given a point $x$ in $\R^d$, by rounding one can find in $O(d)$ time the grid point $x'$ closest to $x$, and $\left\Vert x-x'\right\Vert_{2}\le\frac{\eps r}{2}$. Let $G(x,R)$ denote the set of grid points that are contained in $B^d_2(x,R)$.

\begin{corollary}\label{cor:gridDFD}
	$|G(x,(1+\eps)r)|=O(\frac{1}{\eps})^{d}$.
\end{corollary}
\begin{proof}
	We scale our grid so that the edge length is 1, hence we are looking for the number of lattice points in $B^d_2(x,\frac{1+\eps}{\eps}\sqrt{d})$. By~\Cref{lem:lattice} we get that this number is bounded by the volume of the $d$-dimensional ball of radius $\frac{1+\eps}{\eps}\sqrt{d}+\sqrt{d}\le\frac{3\sqrt{d}}{\eps}$.
	Using Stirling's formula we conclude that
	\[
	V^d_2\left(\frac{3\sqrt{d}}{\eps}\right)
	= \frac{\pi^{\frac{d}{2}}}{\Gamma(\frac{d}{2}+1)}\cdot\left(\frac{3\sqrt{d}}{\eps}\right)^{d}
	\le \left(\frac{\alpha}{\eps}\right)^{d}\, ,
	\]
	where $\alpha$ is a constant.
	For example, if $d$ is even, then
	\[
	V^d_2\left(\frac{3\sqrt{d}}{\eps}\right)
	= \frac{\pi^{\frac{d}{2}}}{(\frac{d}{2})!}\cdot\left(\frac{3\sqrt{d}}{\eps}\right)^{d} 
	\le \frac{\pi^{\frac{d}{2}}}{\sqrt{2\pi}(d/2)^{d/2 + 1/2} e^{-d/2}}\cdot\left(\frac{3\sqrt{d}}{\eps}\right)^{d}
	\le \left(\frac{12.4}{\eps}\right)^{d} = O\left(\frac{1}{\eps}\right)^d\, .
	\]
\end{proof}

Denote by $p_j^i$ the $j$'th point of $C_i$, and let $G_i=\bigcup_{1\le j\le m} G(p_j^i,(1+\eps)r)$ and $\G=\bigcup_{1\le i\le n} G_i$, then by the above corollary we have $|G_i|=m\cdot O(\frac{1}{\eps})^{d}$ and $|\G|=mn\cdot O(\frac{1}{\eps})^{d}$. 
Let $\I_i$ be the set of all curves $\overline{Q}=(x_1,x_2,\dots,x_m)$ with points from $G_i$, such that $\dfd(C_i,\overline{Q})\le (1+\frac{\eps}{2})r$.
\begin{claim}\label{clm:curvesDFD}
	$|\I_i|=O(\frac{1}{\eps})^{md}$ and it can be computed in $O(\frac{1}{\eps})^{md}$ time.
\end{claim}
\begin{proof}
	Let $\overline{Q}\in \I_i$ and let $\tau$ be an alignment with $\sigma_{dF}(\tau(C_i,\overline{Q}))\le (1+\frac{\eps}{2})r$. For each $1\le k\le m$ let $j_k$ be the smallest index such that $(j_k,k)\in\tau$. In other words, $j_k$ is the smallest index that is matched to $k$ by the alignment $\tau$. Since $\dfd(C_i,\overline{Q})\le (1+\frac{\eps}{2})r$, we have $x_k\in B^d_2(p^i_{j_k},(1+\frac{\eps}{2})r)$, for $k=1,\ldots,m$. 
	This means that for any curve $\overline{Q}\in \I_i$ such that $\sigma_{dF}(\tau(C_i,\overline{Q}))\le (1+\frac{\eps}{2})r$, we have $x_k \in G(p^i_{j_k},(1+\frac{\eps}{2})r)$, for $k=1,\ldots,m$.
	By \Cref{cor:gridDFD}, the number of ways to choose a grid point $x_k$ from $G(p^i_{j_k},(1+\frac{\eps}{2})r)$ is bounded by $O(\frac{1}{\eps})^{d}$. 
	
	We conclude that given an alignment $\tau$, the number of curves $\overline{Q}$ with $m$ points from $G_i$ such that $\sigma_{dF}(\tau(C_i,\overline{Q}))\le (1+\frac{\eps}{2})r$ is bounded by $O(\frac{1}{\eps})^{md}$. Finally, by \Cref{lem:alignments}, the total number of curves in $\I_i$ is bounded by $2^{2m}\cdot O(\frac{1}{\eps})^{md}= O(\frac{1}{\eps})^{md}$. 
	
	To construct $\I_i$ we compute, for each of the $O(\frac{1}{\eps})^{md}$ candidates, its discrete \frechet~distance to $C_i$. Thus, we construct $\I_i$ in total time $O(\frac{1}{\eps})^{md} \cdot O(m^2)= O(\frac{1}{\eps})^{md}$. (The latter equality is true, since clearly $(\frac{\alpha}{\eps})^{md} \cdot O(m^2) \le (\frac{c\alpha}{\eps})^{md}$, i.e., $O(m^2) \le c^{md}$, where $\alpha$ is the constant from \Cref{cor:gridDFD} and $c > 1$ is a sufficiently large constant.)  	
\end{proof}

\paragraph{The data structure.} 
Denote $\I=\bigcup_{1\le i\le n} \I_i$, so $|\I|\le n\cdot O(\frac{1}{\eps})^{md}$ and we construct $\I$ in total time $n\cdot O(\frac{1}{\eps})^{md}$.
Next, we would like to store the set $\I$ in a dictionary (a hash table or a lookup table) $\D$, such that given a query curve $Q$, one can find $Q$ in $\D$ (if it exists) in $O(md)$ time. 
We use Cuckoo Hashing~\cite{PR04} to construct a (dynamic) dictionary of linear space, constant worst-case query and deletion time, and constant expected amortized insertion time.
We insert the curves of $\I$ into the dictionary $\D$ as follows. For each $1\le i\le n$ and curve $\overline{Q}\in \I_i$, if $\overline{Q}\notin \D$, insert $\overline{Q}$ into $\D$, and set $C(\overline{Q})\leftarrow C_i$.
The storage space required for $\D$ is $O(|\I|)$, and to construct it we perform $|\I|$ insertions and look-up operations which take in total $O(|\I|\cdot md)=O(|\I|)$ expected time.

\paragraph{A deterministic construction using a prefix tree.} 
Another way to implement the dictionary, which is also dynamic, simple, and does not require randomization at all, is using a binary search tree. Assuming that comparing two curves (given their binary representations) requires $O(md)$ time, the query time will be $O(md\log |\I|)=O((md)^2\log\frac{n}{\eps})$.

Since there is a relatively small number of possible vertices (all the vertices are points of the grid $\G$) we can improve the query time to $O(md\log(\frac{nmd}{\eps}))$ by using a prefix tree instead of a search tree. For details, see~\Cref{appendix:deterministic-construction}.

\paragraph{The query algorithm.}
Let $Q=(q_1,\dots,q_m)$ be the query curve. The query algorithm is as follows: For each $1\le k\le m$ find the grid point $q'_k$ (not necessarily from $\G$) closest to $q_k$. This can be done in $O(md)$ time by rounding. Then, search for the curve $Q'=(q'_1,\dots,q'_m)$ in the dictionary $\D$. If $Q'$ is in $\D$, return $C(Q')$, otherwise, return NO. The total query time is then $O(md)$.

\paragraph{Correctness.}
Consider a query curve $Q=(q_1,\dots,q_m)$. Assume that there exists a curve $C_i\in \C$ such that $\dfd(C_i,Q)\le r$. We show that the query algorithm returns a curve $C^*$ with $\dfd(C^*,Q)\le (1+\eps)r$.

Consider a point $q_k\in Q$. Denote by $q'_k\in \G$ the grid point closest to $q_k$, and let $Q'=(q'_1,\dots,q'_m)$.
We have $\left\Vert q_k-q'_k\right\Vert_{2}\le\frac{\eps r}{2}$, so $\dfd(Q,Q')\le \frac{\eps r}{2}$. By the triangle inequality, 
\[
\dfd(C_i,Q')\le\dfd(C_i,Q)+\dfd(Q,Q')\le r+\frac{\eps r}{2}=(1+\frac{\eps}{2})r,
\]
so $Q'$ is in $\I_i\subseteq \I$. This means that $\D$ contains $Q'$ with a curve $C(Q')\in \C$ such that $\dfd(C(Q'),Q')\le(1+\frac{\eps}{2})r$, and the query algorithm returns $C(Q')$.
Now, again by the triangle inequality,
\[
\dfd(C(Q'),Q)\le\dfd(C(Q'),Q')+\dfd(Q',Q)\le (1+\frac{\eps}{2})r+\frac{\eps r}{2}=(1+\eps)r.
\]

We obtain the following theorem.
\begin{theorem}\label{thm:DFD}
	There exists a data structure for the $(1+\eps,r)$-ANNC under DFD, with $n\cdot O(\frac{1}{\eps})^{md}$ space, $n\cdot O(\frac{1}{\eps})^{md}$ expected preprocessing time, and $O(md)$ query time.
\end{theorem}

\begin{table}[h]
	\begin{center}
		\begin{tabular}{ | l | l | l | l | p{2.5cm} |}
			\hline
			$m$ & Reference & Space & Query & Approx. \\ \hline\hline
			
			\multirow{3}{*}{$\log n$} & \cite{DS17} & $O(n^{4d+1}\log n)$ & $\tilde{O}(n^{4d})$ & $d\sqrt{d}$ \\ \cline{2-5}
			
			& \cite{EP20} & 
			$n^{\Omega(d\log n)}$
			& $\tilde{O}(dn^4)$ & $1+\eps$  
			\\ \cline{2-5}
			
			& \Cref{thm:DFD} & $n^{O(d)}$ & $O(d\log n)$ & $1+\eps$ \\ \hline\hline
			
			\multirow{3}{*}{$O(1)$} & \cite{DS17} & $2^{O(d)}n\log n$ & $2^{O(d)}\cdot \log n$ & $d\sqrt{d}$ \\ \cline{2-5}
			
			& \cite{EP20} & $d^{O(d)}\tilde{O}(n)$ & $O(d\log n)$ & $1+\eps$  \\ \cline{2-5}
			
			& \Cref{thm:DFD} & $2^{O(d)}n$ & $O(d)$ & $1+\eps$ \\ \hline
		\end{tabular}
		\caption{\small Comparing our ANN data structure to previous structures, for a fixed $\eps$ (say $\eps=1/2$).\label{tab:fixM}}
		\vspace{-10pt}
	\end{center}\vspace{-10pt}
\end{table}

\section{The asymmetric setting under DFD}\label{sec:asymmetric-DFD}
In this section, we show how to easily adapt our data structure to the asymmetric setting, by using simplifications of length at most $k$ instead of the original input curves.

Bereg et al.~\cite{BJWYZ08} showed that given a curve $C$ consisting of $m$ points in 3D, and a parameter $r>0$, there is an algorithm that runs in $O(m\log m)$ time and returns a simplification $\Pi$ with minimum number of vertices such that $\dfd(C,\Pi)\le r$. Their algorithm generalizes to higher dimensions, using an approximation algorithm for the minimum enclosing ball problem (see Kumar et al.~\cite{KMY03}). In this section, we use the following generalization of their original approach (\cite{BJWYZ08}, Theorem 1). More details are given in~\Cref{appendix:simplification}.
\begin{lemma}\label{lem:FreshetSimplification}
	Let $C$ be a curve consisting of $m$ points in $\mathbb{R}^d$. Given parameters $k\le m$, $r>0$, and $\eps\in(0,1]$, there is an algorithm that runs in $O\left(\frac{d\cdot m\log m}{\eps}+m\cdot\poly\frac1\eps\right)$ time that either returns 
	a simplification $\Pi$ consisting of $k$ points such that $\dfd(C,\Pi)\le (1+\eps)r$, or declares that for every simplification $\Pi$ with $k$ points, it holds that $\dfd(C,\Pi)>r$.
\end{lemma}

For each $C_i\in\C$, using \Cref{lem:FreshetSimplification} with parameter $\eps=1$, we find a curve $\Pi_i$ of length $k$ such that $\dfd(C_i,\Pi_i)\le 2r$.
If we fail to find such a curve, then we can ignore $C_i$, because it means that $\dfd(Q,C_i)> r$ for any curve $Q$ of length $k$.

To reduce the space consumption of our data structure, we only store candidate curves of length $k$ that are close enough to the simplifications $\Pi_i$. However, since the distance between the simplification $\Pi_i$ and the input curve $C_i$ could be up to $2r$, storing the answers for the set of candidate curves that are within distance $(1+\frac{\eps}{2})r$ from $\Pi_i$ is not enough, because a query $Q$ that is within distance $(1+\eps)r$ from $C_i$ might be as far as $(3+\eps)r$ from $\Pi_i$. Thus, instead, we insert into our data structure all the curves that are within distance $4r$ from $\Pi_i$. This allows us to capture all query curves that are within distance $r$ from $C_i$.

\paragraph{The data structure.} 
We construct our data structure for the original (symmetric) version, with the following modifications. The set of input curves is $\P=\{\Pi_1,\dots,\Pi_n\}$ (instead of $\C$), and the radius parameter is $4r$ (instead of $r$), but the grid edge length remains $\frac{\eps r}{\sqrt{d}}$. 
In addition, we let $\I'_i$ be the set of all curves $\overline{Q}$ with $k$ points from $G_i$, such that $\dfd(\overline{Q},\Pi_i)\le 4r$, and $\I_i$ will be the set of all curves $\overline{Q}\in\I'_i$ such that $\dfd(\overline{Q},C_i)\le (1+\frac{\eps}{2})r$.
We insert the curves in $\I_i$ into the database $\D$ as before: For each $\overline{Q} \in \I_i$, if $\overline{Q} \notin \D$, insert $\overline{Q}$ into $\D$ and set $C(\overline{Q}) \leftarrow C_i$.

Notice that using $4r$ instead of $r$, increases the ratio between the radius and the grid edge length by only a factor of 4, and therefore the bound on $|\I'_i|$ does not change, except that $m$ is replaced by $k$. 
Therefore, the bounds on the storage space and query time are similar to those of the original data structure, where $m$ is replaced by $k$. Thus, the storage space is in $n \cdot O(\frac{1}{\eps})^{kd}$ and the query time is in $O(kd)$. As for the preprocessing time, we get an additional term of $O(nmd\log m)$ for computing the simplifications $\Pi_1,\ldots,\Pi_n$. We also need to compute the distances $\dfd(C_i,\overline{Q})$ in the construction of $\I_i$, for $1 \le i \le n$, which takes $n \cdot O(\frac{1}{\eps})^{kd} \cdot O(mkd)=nm \cdot O(\frac{1}{\eps})^{kd}$ time in total (as $kd\le 2^{kd}$). Thus the total expected preprocessing time is $O(nmd\log m)+nm\cdot O(\frac{1}{\eps})^{kd}=nm\cdot\left(O(d\log m)+O(\frac{1}{\eps})^{kd}\right)$.

\paragraph{Correctness.} Consider a query curve $Q$, and assume that there exists a curve $C_i \in \C$ such that $\dfd(C_i,Q) \le r$. Then, $\Pi_i$ is a curve of length $k$ and $\dfd(C_i,\Pi_i) \le 2r$.
As in the previous section, let $Q'$ be the curve computed by the query algorithm, then $\dfd(Q',Q) \le \frac{\eps r}{2}$. By the triangle inequality, we have $\dfd(Q',C_i)\le\dfd(Q',Q)+\dfd(Q,C_i)\le(1+\frac{\eps}{2})r$, and
\[
\dfd(Q',\Pi_i)\le\dfd(Q',C_i)+\dfd(C_i,\Pi_i) \le (1+\frac{\eps}{2})r+2r\le 4r.
\]
Therefore our data structure contains $Q'$, and the query algorithm returns $C(Q')$, where
$\dfd(C(Q'),Q')\le(1+\frac{\eps}{2})r$. 
Finally, again by the triangle inequality, we have
\[
\dfd(C(Q'),Q)\le\dfd(C(Q'),Q')+\dfd(Q',Q)\le (1+\frac{\eps}{2})r+\frac{\eps r}{2}=(1+\eps)r.
\]

We obtain the following theorem.
\begin{theorem}\label{thm:DFDassym}
	There exists a data structure for the asymmetric $(1+\eps,r)$-ANNC under DFD, with $n\cdot O(\frac{1}{\eps})^{dk}$ space, $nm\cdot \left(O(d\log m)+O(\frac{1}{\eps})^{kd}\right)$ expected preprocessing time, and $O(kd)$ query time.
\end{theorem}

\section{$\ell_{p,2}$-distance of polygonal curves}\label{sec:l_p2-distance}
For the near-neighbor problem under the $\ell_{p,2}$-distance, we use the same basic approach as in \Cref{sec:frechet}, but with two small modifications. The first is that we set the grid's edge length to $\frac{\eps r}{(2m)^{1/p}\sqrt{d}}$, and redefine $G(x,R)$, $G_i$, and $\G$, as in \Cref{sec:frechet} but with respect to the new edge length of our grid.
The second modification is that we redefine $\I_i$ to be the set of all curves $\overline{Q}=(x_1,x_2,\dots,x_m)$ with points from $\G$, such that $\dpt(C_i,\overline{Q})\le (1+\frac{\eps}{2})r$.

We assume without loss of generality from now and to the end of this section that $r=1$ (we can simply scale the entire space by $1/r$), so the grid's edge length is $\frac{\eps}{(2m)^{1/p}\sqrt{d}}$. 
The following corollary is respective to \Cref*{cor:gridDFD}.
\begin{corollary}\label{cor:grid}
	$|G(x,R)|=O\left(1+\frac{m^{1/p}}{\eps}R\right)^d$.
\end{corollary}
\begin{proof}
	We scale our grid so that the edge length is 1, hence we are looking for the number of lattice points in $B^d_2(x,\frac{(2m)^{1/p}\sqrt{d}}{\eps}R)$. By~\Cref{lem:lattice} we get that this number is bounded by the volume of the $d$-dimensional ball of radius $(1+\frac{(2m)^{1/p}}{\eps}R)\sqrt{d}$.
	Using Stirling's formula we conclude,
	\[
	V_2^{d}\left(\left(1+\frac{(2m)^{1/p}}{\eps}R\right)\sqrt{d}\right) =\frac{\pi^{\frac{d}{2}}}{\Gamma(\frac{d}{2}+1)}\cdot\left(\left(1+\frac{(2m)^{1/p}}{\eps}R\right)\sqrt{d}\right)^{d}
	=\alpha^d\cdot\left(1+\frac{m^{1/p}}{\eps}R\right)^{d}
	\]	
	where $\alpha$ is a constant (approximately $4.13\cdot 2^{1/p}$).
\end{proof}

In the following claim we bound the size of $\I_i$, which, surprisingly, is independent of $p$.
\begin{claim}\label{clm:curves}
	$|\I_i|=O(\frac{1}{\eps})^{m(d+1)}$ and it can be computed in $O(\frac{1}{\eps})^{m(d+1)}$ time.
\end{claim}
\begin{proof}
	Let $\overline{Q}=(x_1,x_2,\dots,x_m)\in \I_i$, and let $\tau$ be an alignment with $\sigma_{p,2}(\tau(C_i,\overline{Q}))\le (1+\frac{\eps}{2})$. For each $1\le k\le m$ let $j_k$ be the smallest index such that $(j_k,k)\in\tau$. In other words, $j_k$ is the smallest index that is matched to $k$ by the alignment $\tau$. 
	
	Set $R_k=\|x_k-p_{j_k}^i\|_2$, then we have $\|(R_1,\dots,R_m)\|_p \le \sigma_{p,2}(\tau(C_i,\overline{Q}))\le (1+\frac{\eps}{2})$.
	
	Let $\alpha_k=\left\lceil \frac{m^{1/p}}{\eps}R_k\right\rceil$. By triangle inequality,
	\begin{align*}
	\|(\alpha_1,\alpha_2,\dots,\alpha_m)\|_p &\le \frac{m^{1/p}}{\eps}\|(R_1,R_2,\dots,R_m)\|_p +m^{1/p}\\
	&\le \frac{m^{1/p}}{\eps}\left(1+\frac{\eps}{2}\right)+m^{1/p}< \left(2+\frac{1}{\eps}\right)m^{1/p}.
	\end{align*}
	Clearly, $x_k\in B^d_2(p^i_{j_k},\alpha_k\frac{\eps}{m^{1/p}})$.
	
	We conclude that for each curve $\overline{Q}=(x_1,x_2,\dots,x_m)\in \I_i$ there exists an alignment $\tau$ such that $\sigma_{p,2}(\tau(C_i,\overline{Q}))\le 1+\frac{\eps}{2}$, and a sequence of integers $(\alpha_1,\dots,\alpha_m)$ such that $\|(\alpha_1,\alpha_2,\dots,\alpha_m)\|_p \le (2+\frac{1}{\eps})m^{1/p}$ and $x_k\in B^d_2(p^i_{j_k},\alpha_k\frac{\eps}{m^{1/p}})$, for $k=1,\ldots,m$. Therefore, the number of curves in $\I_i$ is bounded by the multiplication of three numbers:
	
	\begin{enumerate}
		\item The number of alignments that can determine the distance, which is at most $2^{2m}$ by \Cref{lem:alignments}.
		\item \sloppy The number of ways to choose a sequence of $m$ positive integers $\alpha_1,\dots,\alpha_m$ such that $\|(\alpha_1,\alpha_2,\dots,\alpha_m)\|_p \le (2+\frac{1}{\eps})m^{1/p}$, which is bounded by the number of lattice points in $B^m_p((2+\frac{1}{\eps})m^{1/p})$ (the $m$-dimensional $\ell_p$-ball of radius $(2+\frac{1}{\eps})m^{1/p}$). By \Cref{lem:lattice}, this number is bounded by  
		\[
		V^m_p((2+\frac{1}{\eps})m^{1/p}+m^{1/p})\le V^m_p(\frac{4m^{1/p}}{\eps})= \frac{2^m\Gamma(1+1/p)^{m}}{\Gamma(1+m/p)}\left( \frac{4m^{1/p}}{\eps}\right) ^m
		=O(\frac{1}{\eps})^{m}~,
		\] 
		where the last equality follows as $\frac{m^{m/p}}{\Gamma(1+m/p)}=O(1)^m$.
		\item The number of ways to choose a curve $(x_1,x_2,\ldots,x_m)$, such that $x_k  \in G(p^i_{j_k},\alpha_k\frac{\eps}{m^{1/p}})$, for $k=1,\ldots,m$.
		By \Cref{cor:grid}, the number of grid points in $G(p^i_{j_k},\alpha_k\frac{\eps}{m^{1/p}})$ is $O(1+\alpha_k)^d$, so the number of ways to choose $(x_1,x_2,\dots,x_m)$ is at most $\Pi_{k=1}^{m} O(1+\alpha_k)^d=O(1)^{md}\left( \Pi_{k=1}^{m} (1+\alpha_k)\right) ^d$.
		By the inequality of arithmetic and geometric means we have
		\begin{align*}
		\left(\Pi_{k=1}^{m} (1+\alpha_k)^p\right)^{1/p} & \le \left(\frac{\sum^{m}_{k=1}(1+\alpha_k)^p}{m} \right)^{m/p}\\
		& =  \left(\frac{\|(1+\alpha_1,\dots,1+\alpha_m)\|_p}{m^{1/p}}\right)^m \\ 
		& \le \left(\frac{\|1\|_p+\|(\alpha_1,\dots,\alpha_m)\|_p}{m^{1/p}}\right)^m \\
		& \le \left(\frac{m^{1/p}+(2+\frac{1}{\eps})m^{1/p}}{m^{1/p}}\right)^m = O(\frac{1}{\eps})^m,
		\end{align*}
		so $\Pi_{k=1}^{m} O(1+\alpha_k)^d = O(1)^{md}O(\frac{1}{\eps})^{md} = O(\frac{1}{\eps})^{md}$.
	\end{enumerate}
	Finally, $|\I_i|\le 2^{2m}\cdot O(\frac{1}{\eps})^m\cdot O(\frac{1}{\eps})^{md}\le O(\frac{1}{\eps})^{m(d+1)}$.
\end{proof}

The data structure and query algorithm are similar to those we described for DFD, and the size of $\I_i$ and $\I$ is roughly the same (here there is an additional $O(\frac{1}{\eps})^m$ factor in the space bound). Therefore, the query time, storage space, and preprocessing time are roughly similar, but we still need to show that the algorithm is correct.

\paragraph{Correctness.}
Consider a query curve $Q=(q_1,\dots,q_m)$. Assume that there exists a curve $C_i\in \C$ such that $\dpt(C_i,Q)\le 1$. We will show that the query algorithm returns a curve $C^*$ with $\dpt(C^*,Q)\le 1+\eps$.

Consider a point $q_k\in Q$. Denote by $q'_k\in \G$ the grid point closest to $q_k$, and let $Q'=(q'_1,\dots,q'_m)$.
We have $\|q_k-q'_k\|_{2}\le\frac{\eps}{2(2m)^{1/p}}$. 
Let $\tau$ be an alignment such that the $\ell_{p,2}$-cost of $\tau$ w.r.t. $C_i$ and $Q$ is at most $1$.
Unlike the \frechet\ distance, $\ell_{p,2}$-distance for curves does not satisfy the triangle inequality. However, by the triangle inequality under $\ell_2$ and $\ell_p$, we get that 
the $\ell_{p,2}$-cost of $\tau$ w.r.t. $C_i$ and $Q'$ is 
\begin{align*}
\sigma_{p,2}(\tau(C_i,Q'))&=\left(\sum_{(j,t)\in\tau}\|p^i_j-q'_t\|_2^p\right)^{1/p}\le\left(\sum_{(j,t)\in\tau}\left(\|p^i_j-q_t\|_2+\|q_t-q'_t\|_2\right)^p\right)^{1/p}\\
&\le\left(\sum_{(j,t)\in\tau}\|p^i_j-q_t\|_2^p\right)^{1/p}+\left(\sum_{(j,t)\in\tau}\|q_t-q'_t\|_2^p\right)^{1/p}\\
&\le 1+\left(2m\left(\frac{\eps}{2(2m)^{1/p}}\right)^p\right)^{1/p}\le 1+\frac{\eps}{2}.
\end{align*}

So $\dpt(C_i,Q')\le 1+\frac{\eps}{2}$, and thus $Q'$ is in $I_i\subseteq \I$. This means that $\T$ contains $Q'$ with a curve $C(Q')\in \C$ such that $\dpt(C(Q'),Q')\le 1+\frac{\eps}{2}$, and the query algorithm returns $C(Q')$.
Now, again by the same argument (using an alignment with $\ell_{p,2}$-cost at most $1+\frac{\eps}{2}$ w.r.t. $C(Q')$ and $Q'$), we get that $\dpt(C(Q'),Q)\le 1+\frac{\eps}{2}+\left(2m\left(\frac{\eps}{2(2m)^{1/p}}\right)^p\right)^{1/p}=1+\eps$.

We obtain the following theorem.
\begin{theorem}\label{thm:lp2}
	There exists a data structure for the $(1+\eps,r)$-ANNC under $\ell_{p,2}$-distance, with $n\cdot O(\frac{1}{\eps})^{m(d+1)}$ space, $n \cdot O(\frac{1}{\eps})^{m(d+1)}$ expected preprocessing time, and $O(md)$ query time.
\end{theorem}

As mentioned in the preliminaries section, the DTW distance between two curves equals to their $\ell_{1,2}$-distance, and therefore we obtain the following theorem.
\begin{theorem}\label{thm:DTW}
	There exists a data structure for the $(1+\eps,r)$-ANNC under DTW, with $n\cdot O(\frac{1}{\eps})^{m(d+1)}$ space, $n\cdot O(\frac{1}{\eps})^{m(d+1)}$ expected preprocessing time, and $O(md)$ query time.
\end{theorem}

\section{The asymmetric setting under $\ell_{p,2}$-distance}\label{sec:assym-l_p2-distance}
In~\Cref{sec:asymmetric-DFD}, we strongly rely on the fact that DFD satisfies the triangle inequality, in order to provide a data structure with storage space independent of $m$. However, the general $\ell_{p,2}$ distance does not satisfy the triangle inequality, not even up to a constant factor. Lemire~\cite{Lemire09} proved the following weak version of the triangle inequality for $\ell_{p,2}$ distance, and showed it to be tight:
For any three curves $A,B,C$ of length $k$, 
it holds that $\dpt(A,B)\le k^{\nicefrac{1}{p}}\left(\dpt(A,C)+\dpt(C,B)\right)$.

Nonetheless, in~\Cref{sec:l_p2-distance} we showed how to apply our approach to ANNC under the general $\ell_{p,2}$-distance. This was possible due to the fact that we always match $Q'$ to $Q$ in a ``one-to-one'' alignment, a special case where the triangle inequality does hold. In our solution to the asymmetric case, we use the triangle inequality between $\Pi_i$, $C_i$, and $Q$, in which case we do not have a one-to-one matching. Therefore, our analysis for the asymmetric \frechet\ distance does not trivially apply to the $\ell_{p,2}$-distance in general.
In this section we show how we can still adapt our algorithm to the asymmetric ANNC under $\ell_{p,2}$ distance. The storage space and query time is almost the same as in \Cref{thm:lp2}, where we replace $m$ by $k$.

We begin by describing some important characterizations of curve alignments. 
First, notice that a curve alignment $\tau=\langle(i_1,j_1),\dotsc,(i_t,j_t)\rangle$ can be viewed as a bipartite graph on the vertices of the two curves. Moreover, we can assume that each connected component in this graph is a star graph, i.e., a single vertex from the first curve is connected to one or more vertices from the second curve, and vice versa. Indeed, as we claimed in the proof of~\Cref{lem:alignments}, if there exist pairs (or edges) $(i_k,j_k-1),(i_k,j_k),(i_k+1,j_k)$, then removing $(i_k,j_k)$ from $\tau$ results in a legal curve alignment with a smaller (or equal) cost. 
The following special type of curve alignment is crucial in the construction and proof of our algorithm.

\paragraph{One-way alignment.} An alignment $\tau = \langle(i_1,j_1),\dotsc,(i_t,j_t)\rangle$ is a \dfn{one-way alignment} if for any $1\le s\le t$, we have $i_s=s$. In other words, in the view of $\tau$ as a bipartite graph, we get a set of stars such that all the centers are in the second curve, and thus each index of the first curve appears in exactly one pair of $\tau$.

\begin{remark}
	A nice property of one-way alignments is that, while the triangle inequality does not apply for $\ell_{p,2}$ distance between curves in general, it does hold when the $\ell_{p,2}$ distances between the curves are obtained by one-way alignments. See \Cref{clm:oneWayTriangleEnq} in \Cref{appendix:one-way} for details.
\end{remark}

\subsection{Simplification under $\ell_{p,2}$-distance}
As for DFD, in order to construct a data structure with storage space independent of $m$, we will need to compute simplifications of length $k$ for the input curves. The points in a simplification can be arbitrary, but for the purpose of this section it is enough to use vertex-restricted simplifications. A simplification $\Pi$ of a curve $C$ is \emph{vertex-restricted} if the points of $\Pi$ are from the vertices of $C$, and follow the same ordering.

\begin{lemma}\label{lem:simp-approx}
	Let $C=(p_1,\dots,p_m)$ be a curve consisting of $m$ points in $\mathbb{R}^d$.
	Denote by $\overline{\Pi}=(x_1,\dots,x_k)$ the closest simplification to $C$ with $k$ points under the $\ell_{p,2}$-distance.
	Then there exist a vertex-restricted simplification $\Pi$ of $C$ with $k$ points such that $\dpt(C,\Pi)\le 2\cdot\dpt(C,\overline{\Pi})$.
\end{lemma}
\begin{proof}
	Let $\tau$ be an alignment such that $\dpt(C,\overline{\Pi})=\sigma_{p,2}(\tau(C,\overline{\Pi}))$. We can assume that $\tau$ is a one-way alignment, as otherwise we can remove points from $\overline{\Pi}$ without increasing the distance to $C$.
	For $1\le j \le k$, let $A_j=\{i\in[m]\mid (i,j)\in \tau\}$ be the set of indices matched to $j$ by $\tau$. Let $i_j=arg\min_{t\in A_j}\|p_{t}-x_j\|_2$. Note that for every $i\in A_j$ it holds that $\|p_{i}-p_{i_{j}}\|_{2}\le\|p_{i}-x_{j}\|_{2}+\|x_{j}-p_{i_{j}}\|_{2}\le2\|p_{i}-x_{j}\|_{2}$.
	Set $\Pi=(p_{i_1},\dots,p_{i_k})$, so $\Pi$ is a vertex-restricted simplification of $C$ consisting of $k$ points from $C$. It holds that
	\[
	\sigma_{p,2}(\tau(C,\overline{\Pi}))=
	\left(\sum_{j=1}^{k}\sum_{i\in A_{j}}\|p_{i}-x_{j}\|_{2}^{p}\right)^{1/p}\ge
	\left(\sum_{j=1}^{k}\sum_{i\in A_{j}}\frac{1}{2^{\phantom{.\hspace{-1pt}}^p}}\|p_{i}-p_{i_{j}}\|_{2}^{p}\right)^{1/p}=
	\frac{1}{2}\sigma_{p,2}(\tau(C,\Pi))~.
	\]
	The lemma follows as $\dpt(C,\Pi)\le \sigma_{p,2}(\tau(C,\Pi))\le 2\sigma_{p,2}(\tau(C,\overline{\Pi}))=2\dpt(C,\overline{\Pi})$.
\end{proof}
\begin{lemma}\label{lem:simp-compute}
	Given a curve $C=(p_1,\dots,p_m)$ consisting of $m$ points in $\mathbb{R}^d$, and parameters $k<m$, $p\ge 1$, there exists an algorithm that runs in $O(m^3k+m^2d)$ time and computes a vertex-restricted simplification $\Pi$ of $C$ with $k$ points, such that $\dpt(C,\Pi)$ is minimized.
\end{lemma}
\begin{proof}
	We show how to compute $\Pi$ using a dynamic programming technique. We begin by precomputing all the pairwise distances in $C$, such that for every $i,j$, we will have a constant time access to $\Vert p_i-p_j\Vert_2^p$. This takes $O(m^2d)$ time (ignoring the time it takes to compute the $p$-power of a number).
	
	We define $OPT[i,j,x]$ as follows.
	Let $C[1:i]=(p_1,\dots,p_i)$, and let $\Pi_j^x=(p_{i_1},\dots,p_{i_x})$ be a vertex-restricted simplification with $x$ points from $C[1:j]$, such that $i_x=j$, and $\dpt(C[1:i],\Pi_j^x)$ is minimized. Then $OPT[i,j,x]=\left(\dpt(C[1:i],\Pi_j^x)\right)^p$.
	We compute $OPT[i,j,x]$ for any $1\le x \le k$, $x\le j\le m$, and $x\le i \le m$ as follows.
	
	First, for any $1\le j\le m$ we have $OPT[1,j,1]=\|p_1-p_j\|_2^p$, and for any $2\le i \le m$ we have $OPT[i,j,1]=OPT[i-1,j,1]+\|p_i-p_j\|_2^p$. Thus $OPT[i,j,1]$ can be computed in $O(m^2)$ time for all $1\le i\le m$, $1\le j \le m$.
	
	Next, we compute $OPT[i,j,x]$ for all $2\le x \le k$, $x\le j\le m$, and $x\le i \le m$, using the following formula, in $O(m^3k)$ time:
	\[
	OPT[i,j,x]=\|p_i-p_j\|_2^p+\min \left\{ \min_{x-1\le j'< j}OPT[i-1,j',x-1],~ OPT[i-1,j,x] \right\}.
	\]
	Any alignment w.r.t.~$C[1:i]$ and a simplification $\Pi_j^x$ has to match $p_i$ and $p_j$. Let $\tau$ be an alignment such that $\sigma_{p,2}(\tau(C[1:i],\Pi_j^x))$ is minimized. If $\tau$ matches $p_j$ to $p_{i-1}$ then $OPT[i,j,x]=\|p_i-p_j\|_2^p + OPT[i-1,j,x]$. Otherwise, there exists some $x-1\le j'< j$ such that $\tau$ matches $p_{j'}$ to $p_{i-1}$, and $OPT[i,j,x]=\|p_i-p_j\|_2^p+OPT[i-1,j',x-1]$.
	
	Finally, we return 
	\[
	\min_{\substack{1\le x\le k\\ x\le j\le m}}OPT[m,j,x],
	\]
	which can be computed in $O(mk)$ time.
	
	Clearly, the vertex-restricted simplification $\Pi$ that minimizes $\dpt(C,\Pi)$, and the corresponding alignment, can be found by backtracking in $O(mk)$ time.
\end{proof}

\subsection{The data structure}\label{subsec:asym_lp}
Our data structure for the asymmetric case under $\ell_{p,2}$-distance is very similar to the asymmetric case under DFD, but, without the triangle inequality, we have to use more involved counting arguments in order to achieve similar bounds. 

Fix some $p>1$ and assume w.l.o.g., as in \Cref{sec:l_p2-distance}, that $r=1$.
For every input curve $C_{i}$, compute the optimal vertex-restricted simplification $\Pi_{i}$ of $C_i$, with at most $k$ points, using \Cref{lem:simp-compute} in $O(m^3 k+m^2d)$ time.
Note that in addition we obtain a one-way alignment $\tau_i$ for which $\dpt(C_i,\Pi_i)= \sigma_{p,2}(\tau_i(C_i,\Pi_i))$.
We can assume that $\dpt(C_i,\Pi_i)\le 2$, as otherwise we can just ignore $C_i$, since by \Cref{lem:simp-approx} we get that for every curve $Q$ of length $\le k$, $\dpt(C_i,Q)>1$. 

We again construct our data structure for the original problem (for $\ell_{p,2}$-distance), but with the following modifications. The set of input curves is $\P=\{\Pi_1,\dots,\Pi_n\}$, and the radius parameter is $7k^{\nicefrac{1}{p}}$ (the length of grid edges is still $\frac{\eps}{(2m)^{1/p}\sqrt{d}}$).
Now let $\I'_i$ be the set of all curves $\overline{Q}$ with $k$ points from $G_i$, such that $\dpt(\overline{Q},\Pi_i)\le 7k^{\nicefrac{1}{p}}$.
In addition, let $\I_i$ be the set of all curves $\overline{Q}\in\I'_i$ such that $\dpt(\overline{Q},C_i)\le 1+\frac{\eps}{2}$.

\paragraph{Correctness.} Consider a query curve $Q=(q_1,\dots,q_k)$ such that $\dpt(Q,C_i)\le 1$, and let $Q'=(q'_1,\dots,q'_k)$ be the grid curve that was computed by the query algorithm. 
Following the same arguments as in \Cref{sec:l_p2-distance}, there exists an alignment $\tau$ such that $\sigma_p(\tau(Q',C_i))\le 1+\frac{\eps}{2}$, and thus $\dpt(Q',C_i)\le 1+\frac{\eps}{2}$. 
We have that
\begin{align*}
\dpt(Q',\Pi_{i}) & \le(2k-1)^{\nicefrac{1}{p}}\cdot\dfd(Q',\Pi_{i})\\
& \le(2k-1)^{\nicefrac{1}{p}}\cdot\left(\dfd(Q',C_{i})+\dfd(C_{i},\Pi_{i})\right)\\
& \le(2k-1)^{\nicefrac{1}{p}}\cdot\left(\dpt(Q',C_{i})+\dpt(C_{i},\Pi_{i})\right)\\
& \le(2k-1)^{\nicefrac{1}{p}}\cdot\left(1+\frac{\eps}{2}+2\right)\cdot\dpt(Q,C_{i})\le7 k^{\nicefrac{1}{p}}~,
\end{align*}
where the first and third inequalities follow by the fact that for every $\vec{x}\in\mathbb{R}^d$, $ \Vert \vec{x}\Vert_\infty\le \Vert \vec{x}\Vert_p\le d^{\nicefrac1p}\cdot \Vert \vec{x}\Vert_\infty$, and the second inequality follows by the triangle inequality of discrete \frechet\ distance. It follows that
$Q'\in \I_i$. The query algorithm then returns a curve $C(Q')$ such that $\dpt(Q',C(Q'))\le 1+\frac{\eps}{2}$, and $\dpt(Q,C(Q'))\le 1+\eps$.

\paragraph{Storage space.} A calculation using arguments similar to those in \Cref{sec:l_p2-distance} (where $\I'_{i}$ is computed using the same algorithm, but with a radius multiplied by $7k^{\nicefrac{1}{p}}$), bounds the storage by $O(\frac{1}{\epsilon})^{k(d+1)}\cdot \left(k\cdot m^{d}\right)^{\frac{k}{p}}$, and the preprocessing time by $nm\cdot \left(O(m^2k+md))+O(\frac{1}{\epsilon})^{k(d+1)}\cdot \left(k\cdot m^{d}\right)^{\frac{k}{p}}\right)$ only (we leave this as an exercise for the diligent reader). Next, we provide a better counting argument and remove the $k^{\nicefrac{1}{p}}$ factor from the base of the exponent.

Fix an input curve $C_i=(p_1,\dots,p_m)$, 
a simplification $\Pi_i$ of length $k$, and a one-way alignment $\tau_i$ such that $\dpt(C_i,\Pi_i)= \sigma_{p,2}(\tau_i(C_i,\Pi_i))\le 2$. We say that a subset $A\subseteq C_i$ is \emph{consecutive} if it is of the form $(p_s,p_{s+1},\dots,p_t)$, for some $s,t\in[m]$. 

\begin{claim}\label{clm:curve-partition}
	The points of $C_i$ can be partitioned into $k'\le 3k$ consecutive subsets $A_1,\dots,A_{k'}$ such that for every $1\le s\le k'$ there exists a center point $y_s$ for which $A_s\subset B_2^d(y_{s},\frac{2}{k^{1/p}})$.
\end{claim}
\begin{proof}
	For any $p_j\in C_i$, denote by $\pi(p_j)$ the single point in $\Pi_i$ that is matched to $p_j$ by $\tau_i$. Let $L=\{p_j\in C_i \mid \Vert p_j-\pi(p_j)\Vert_2 >\frac{2}{k^{1/p}} \}$, then $|L|\le k$, as otherwise we have
	\[
	\sigma_{p,2}(\tau_{i}(C_{i},\Pi_{i}))=\left(\sum_{j=1}^{m}\Vert p_{j}-\pi(p_j)\Vert_{2}^{p}\right)^{1/p}\ge\left(\sum_{p_{j}\in L}\Vert p_{j}-\pi(p_j)\Vert_{2}^{p}\right)^{1/p}>\left(k\cdot\frac{2^{p}}{k}\right)^{1/p}=2~,
	\]
	which is a contradiction.
	
	We construct a set $\A$ of consecutive subsets $A_1,\dots,A_{k'}$ and the corresponding center points $y_1,\dots,y_{k'}$ as follows. 
	First, for every $1\le s'\le k$ let $\tilde{A}_{s'}=\{p_j\in C_i \mid (j,s')\in\tau_i\}$. Then $\tilde{A}_1,\dots,\tilde{A}_{k}$ is a set of consecutive subsets that partition $C_i$. For every set $\tilde{A}_{s'}$, let $\tilde{A}_{s'}\cap L=\{p_{i_1},\dots,p_{i_t}\}$, and insert into $\A$ the partition of $\tilde{A}_{s'}$ into $2t+1$ consecutive subsets $A^1_{s'},\{p_{i_1}\},A^2_{s'},\{p_{i_2}\},\dots,A^t_{s'},\{p_{i_t}\},A^{t+1}_{s'}$.
	The number of subsets in $\A$ is at most $2|L|+k\le 3k$.
	For each subset $A_s\in\A$, if it is a singleton $\{p_j\}$ then we set $y_s=p_j$, otherwise, $A_s=A^\ell_{s'}$ and we set $y_s$ to be the $s'$ point of $\Pi_i$. Clearly, $\Vert p_j-y_s\Vert_2\le \frac{2}{k^{1/p}}$ for every $p_j\in A_s$.
\end{proof}

For each curve $C_i\in\C$ we can compute in $O(m)$ time the set of at most $3k$ consecutive subsets $A_1,\dots,A_{k'}$ and centers $y_1,\dots,y_{k'}$ from~\Cref{clm:curve-partition}.

Now, consider a curve $\overline{Q}=(x_1,\dots,x_k)\in \I_i$, and let $\tau$ be an alignment such that $\sigma_{p,2}(\tau(C_i,\overline{Q}))\le 1+\frac\eps2$.
For every $1\le j\le k$, let $1\le c(j)\le m$ be the smallest index such that $(c(j),j)\in\tau$, and let $1\le f(j)\le k'$ be the index such that $p_{c(j)}\in A_{f(j)}$.
Since the sets $A_1,\dots,A_{k'}$ are consecutive and $\tau$ is an alignment, $f$ is a monotonically non-decreasing function from $[k]$ to $[k']$. 
It follows that the number of possible functions $f$ is equivalent to the number of ways to distribute $k$ identical balls into $k'$ distinct boxes, thus ${k'-1+k \choose k}< {4k \choose k}\le O(2)^{k}$.

Set $R_j=\|x_j-p_{c(j)}\|_2$ and $\alpha_j=\left\lceil \frac{k^{1/p}}{\eps}R_{j}\right\rceil $.
We have $\|(R_1,\dots,R_k)\|_p \le \sigma_{p,2}(\tau(C_i,\overline{Q}))\le 1+\frac\eps2$.
By the triangle inequality,
\begin{align*}
\|(\alpha_{1},\alpha_{2},\dots,\alpha_{k})\|_{p} & \le\Vert(\frac{k^{1/p}}{\eps}R_{1}+1,\dots,\frac{k^{1/p}}{\eps}R_{k}+1)\Vert_{p}\\
& \le\frac{k^{1/p}}{\eps}\|(R_{1},\dots,R_{k})\|_{p}+k^{1/p}\quad~~\le~\frac{k^{1/p}}{\eps}(1+\frac{3}{2}\eps)~.
\end{align*}
The number of ways to choose such a sequence $\alpha_1,\dots,\alpha_k$ of $k$ positive integers is bounded by the number of lattice points in $B^k_p(\frac{k^{1/p}}{\eps}(1+\frac32\eps))$ (the $k$-dimensional $\ell_p$-ball of radius $ \frac{k^{1/p}}{\eps}(1+\frac32\eps)$). By \Cref{lem:lattice}, this number is bounded by $V^k_p( \frac{k^{1/p}}{\eps}(1+\frac32\eps))=O(\frac{1}{\eps})^{k}$.

For every $1\le j\le k$ we have $\Vert p_{c(j)}-y_{f(j)}\Vert_2\le \frac{2}{k^{1/p}}$, and thus
\[
\Vert x_{j}-y_{f(j)}\Vert_{2}\le\Vert x_{j}-p_{c(j)}\Vert_{2}+\Vert p_{c(j)}-y_{f(j)}\Vert_{2}\le\alpha_{j}\frac{\eps}{k^{1/p}}+\frac{2}{k^{1/p}}.
\]
Therefore, once $f$ and $(\alpha_1,\dots,\alpha_k)$ are fixed, it is sufficient to count the number of ways to choose a curve $(x_1,x_2,\ldots,x_k)$ such that $x_j  \in G\left(y_{f(j)},\alpha_j\frac{\eps}{k^{1/p}}+\frac{2}{k^{1/p}}\right)$ for every $1\le j \le k$.

By \Cref{cor:grid}, as we use the $\frac{\eps}{m^{1/p}}$ grid,  the number of grid points in $G\left(y_{c(j)},\alpha_j\frac{\eps}{k^{1/p}}+\frac{2}{k^{1/p}}\right)$ is $O\left(1+\frac{m^{1/p}}{\eps}\left(\alpha_{j}\frac{\eps}{k^{1/p}}+\frac{2}{k^{1/p}}\right)\right)^{d}\le O\left(\frac{m}{k}\right)^{d/p}\cdot\left(\frac{3}{\eps}+\alpha_{j}\right)^{d}$.
Thus the number of ways to choose $(x_1,x_2,\dots,x_k)$ is at most $O\left(\frac{m}{k}\right)^{dk/p}\left(\Pi_{j=1}^{k}\left(\frac{3}{\eps}+\alpha_{j}\right)\right)^{d}$.
By the inequality of arithmetic and geometric means we have
\begin{align*}
\Pi_{j=1}^{k}\left(\frac{3}{\eps}+\alpha_{j}\right) & =\left(\Pi_{j=1}^{k}\left(\frac{3}{\eps}+\alpha_{j}\right)^{p}\right)^{1/p}\\
& \le\left(\frac{\sum_{j=1}^{k}(\frac{3}{\eps}+\alpha_{j})^{p}}{k}\right)^{k/p}\\
& \le\left(\frac{\|(\frac{3}{\eps},\dots,\frac{3}{\eps})\|_{p}+\|(\alpha_{1},\dots,\alpha_{k})\|_{p}}{k^{1/p}}\right)^{k}\\
& \le\left(\frac{\frac{3}{\eps}\cdot k^{1/p}+\frac{k^{1/p}}{\eps}(1+\frac{3}{2}\eps)}{k^{1/p}}\right)^{k}=O\left(\frac{1}{\eps}\right)^{k}~.
\end{align*}
It follows that for some fixed $f$ and $(\alpha_1,\dots,\alpha_k)$, the number of ways to choose $(x_1,\dots,x_k)$ is bounded by $O(\frac{1}{\eps})^{kd}\cdot\left(\frac{m}{k}\right)^{kd/p}$.

We conclude that, 
$|\I_{i}|={O(2)}^{k}\cdot O(\frac{1}{\eps})^{k}\cdot(\frac{1}{\eps})^{kd}\cdot\left(\frac{m}{k}\right)^{kd/p}=O(\frac{1}{\eps})^{k(d+1)}\cdot\left(\frac{m}{k}\right)^{kd/p}$, and thus $|\I|=n\cdot O(\frac{1}{\eps})^{k(d+1)}\cdot\left(\frac{m}{k}\right)^{kd/p}$.

Following the same lines as in the counting argument, an efficient implementation of the preprocessing time is self evident.
\begin{theorem}\label{thm:lp2Assym}
	There exists a data structure for the Asymmetric $(1+\eps,r)$-ANNC under $\dpt$, with $n\cdot O(\frac{1}{\eps})^{k(d+1)}\cdot\left(\frac{m}{k}\right)^{kd/p}$ space, $nm\cdot \left(O(m^2k+md)+O(\frac{1}{\eps})^{k(d+1)}\cdot\left(\frac{m}{k}\right)^{kd/p}\right)$ preprocessing time, and $O(kd)$ query time.
\end{theorem}

\section{Approximate range counting}\label{sec:counting}
In the range counting problem for curves, we are given a set $\C$ of $n$ curves, each consisting of $m$ points in $d$ dimensions, and a distance measure for curves $\delta$. The goal is to preprocess $\C$ into a data structure that given a query curve $Q$ and a threshold value $r$, returns the number of curves that are within distance $r$ from $Q$.

In this section we consider the following approximation version of range counting for curves, in which $r$ is part of the input. Note that by storing pointers to curves instead of just counters, we can obtain a data structure for the approximate range searching problem (at the cost of an additional $O(n)$-factor to the storage space).

\begin{problem}[$(1+\eps,r)$-approximate range-counting for curves]\label{prb:r-range-counting}
	Given a parameter $r$ and $0<\eps\le 1$, preprocess $\C$ into a data structure that given a query curve $Q$, returns the number of all the input curves whose distance to $Q$ is at most $r$ plus possibly additional input curves whose distance to $Q$ is greater than $r$ but at most $(1+\eps)r$.
\end{problem}

We construct the dictionary $\D$ (implemented as a dynamic hash table, or a prefix tree) for the curves in $\I$ as in \Cref{sec:l_p2-distance}, as follows. 
For each $1\le i\le n$ and curve $\overline{Q}\in \I_i$, if $\overline{Q}$ is not in $\D$, insert it into $\D$ and initialize $C(\overline{Q})\leftarrow 1$. Otherwise, if $\overline{Q}$ is in $\D$, update $C(\overline{Q})\leftarrow C(\overline{Q})+1$. Notice that $C(\overline{Q})$ holds the number of curves from $\C$ that are within distance $(1+\frac{\eps}{2})r$ to $\overline{Q}$.
Given a query curve $Q$, we compute $Q'$ as in \Cref{sec:l_p2-distance}. If $Q'$ is in $\D$, we return $C(Q')$, otherwise, we return 0.

Clearly, the storage space, preprocessing time, and query time are similar to those in \Cref{sec:l_p2-distance}. We claim that the query algorithm returns the number of curves from $\C$ that are within distance $r$ to $Q$ plus possibly additional input curves whose distance to $Q$ is greater than $r$ but at most $(1+\eps)r$. Indeed, let $C_i$ be a curve such that $\dpt(C_i,Q)\le r$. As shown in \Cref{sec:l_p2-distance} we get $\dpt(C_i,Q')\le (1+\frac{\eps}{2})r$, so $Q'$ is in $\I_i$ and $C_i$ is counted in $C(Q')$. Now let $C_i$ be a curve such that $\dpt(C_i,Q)>(1+\eps)r$. If $\dpt(C_i,Q')\le (1+\frac{\eps}{2})r$, then by a similar argument (switching the rolls of $Q$ and $Q'$) we get that $\dpt(C_i,Q')\le (1+\eps)r$, a contradiction. So $\dpt(C_i,Q')> (1+\frac{\eps}{2})r$, and thus $C_i$ is not counted in $C(Q')$.

We obtain the following theorem.
\begin{theorem}
	There exists a data structure for the $(1+\eps,r)$-approximate range-counting for curves under $\ell_{p,2}$-distance, with $n \cdot O(\frac{1}{\eps})^{m(d+1)}$ space, $n\log(\frac{n}{\eps})\cdot O(\frac{1}{\eps})^{m(d+1)}$ preprocessing time, and $O(md\log(\frac{nmd}{\eps}))$ query time. (Under DFD, the exponent in the bounds for the space and preprocessing time is $md$ rather than $m(d+1)$.)
\end{theorem}

\section{Simplification in $d$-dimensions}\label{appendix:simplification}
The algorithm of Bereg et al.~\cite{BJWYZ08} receives as input a curve $C$ consisting of $m$ points in $\R^3$, and a parameter $r>0$. In $O(m\log m)$ time, it returns a curve $\Pi$ such that $\dfd(C,\Pi)\le r$, and $\Pi$ has the minimum number of vertices among all curves within distance $r$ from $C$.
The algorithm operates in a greedy manner, by repeatedly executing Megiddo's~\cite{Meg84}  minimum enclosing ball (MEB) algorithm for points in $\R^3$, which takes linear time.

We generalize the algorithm of Bereg et al. for curves in $\R^d$, by using an algorithm presented by Kumar et al.~\cite{KMY03} for approximated minimum enclosing ball (AMEB) in $\R^d$. Formally, given a set $A$ of $n$ points in $\R^d$ and a parameter $\eps\in(0,1]$, the goal is to find an enclosing ball of $A$ with radius $r>0$, where the minimum enclosing ball of $A$ has radius at least $\frac{r}{1+\eps}$. The algorithm of~\cite{KMY03} can find an AMEB in $O(\frac{nd}{\eps}+\eps^{-4.5}\log\frac{1}{\eps})$ time. In particular, given an additional parameter $r>0$, this algorithm either returns an enclosing ball of $A$ with radius $(1+\eps)r$, or declares that the minimum enclosing ball of $A$ has radius larger than $r$.

Next, we describe our modified algorithm. Consider a curve $C=(x_1,\dots,x_m)$, and denote  $C[i,j]=(x_i,\dots,x_j)$. The following sub-procedure takes as input a curve $A$ and returns a point $y$ and an index $s$, such that the ball with radius $(1+\eps)r$ centered at $y$ covers the prefix $A[1,s]$, and (if $s<|A|$) the minimum enclosing ball of $A[1,s+1]$ has radius larger than $r$.
\begin{enumerate}
	\item By iterative probing, using an algorithm for AMEB, find some $t$ such that $A[1,2^t]$ can be covered by a ball of radius $(1+\eps)r$, while $A[1,2^{t+1}]$ cannot be covered by a ball of radius $r$. If all the points in $A$ can be enclosed by a single ball of radius $(1+\eps)r$ centered at $y$, simply return $y$ and $|A|$.
	\item By binary search, again using an algorithm for AMEB, find some $s\in [2^{t},2^{t+1})$ such that $A[1,s]$ can be covered by a ball of radius $(1+\eps)r$, and $A[1,s+1]$ cannot be covered by a ball of radius $r$. Let $y\in\R^d$ be the center of this ball. Return $y$ and $s$.
\end{enumerate}
Starting from the input $A=C[1,m]$, repeat the above sub-procedure such that in each step the input is the suffix of $C$ that was not yet covered by the previous steps (i.e. $A[s+1,m]$). Let $(y_1,\dots,y_q)$ be the sequence of output points.

\Cref{lem:FreshetSimplification} is an easy corollary of the following lemma.
\begin{lemma}
	Let $C$ be a curve consisting of $m$ points in $\R^d$. Given parameters $r>0$, and $\eps\in(0,1]$, the algorithm above runs in $O\left(\frac{d\cdot m\log m}{\eps}+m\cdot\eps^{-4.5}\log\frac{1}{\eps}\right)$ time and returns a curve $\Pi=(y_1,\dots,y_q)$ such that $\dfd(C,\Pi)\le (1+\eps)r$. Furthermore, for every curve $\Pi'$ with less than $q$ points, it holds that $\dfd(C,\Pi')>r$.
\end{lemma}
\begin{proof}[Proof sketch.]
	We start by analyzing the running time for a single iteration of the sub-procedure, when using the algorithm of~\cite{KMY03} to find an AMEB. The total time for the first step of the sub-procedure (finding $t$) is 
	$$\sum_{i=1}^{t+1}O(\frac{2^i\cdot d}{\eps}+\eps^{-4.5}\log\frac{1}{\eps})=O(\frac{2^t\cdot d}{\eps}+t\cdot\eps^{-4.5}\log\frac{1}{\eps}).$$
	In the second step, there are $O(t)$ executions of \cite{KMY03} on a set of size at most $2^{t+1}$, so the total time for this step is $t\cdot O(\frac{2^t\cdot d}{\eps}+\eps^{-4.5}\log\frac{1}{\eps})$. 
	
	Let $m_i$ be the length of the subcurve covered by the point $y_i$ that was found in the $i$'th iteration of the sub-procedure. The total time spent for finding $y_i$ is therefore ${\log m_i\cdot O(\frac{m_i\cdot d}{\eps}+\eps^{-4.5}\log\frac{1}{\eps})}$, and the total running time of the algorithm is 
	
	\[	
	\sum_{i=1}^{q}\log m_{i}\cdot O\left(\frac{m_{i}\cdot d}{\eps}+\eps^{-4.5}\log\frac{1}{\eps}\right)=O\left(\frac{d\cdot m\log m}{\eps}+m\cdot\eps^{-4.5}\log\frac{1}{\eps}\right)~,
	\]
	where we used the concavity of the $\log$ function, and the fact $\sum_{i=1}^{q}m_i=m$.
	
	Next we argue the correctness. Clearly, $\dfd(C,\Pi)\le (1+\eps)r$. Let $s_0=0,s_1,\dots,s_q=m$ be the sequence of indices (of vertices in $C$) found during the execution of the algorithm, such that the ball of radius $(1+\eps)r$ around $y_i$ covers $C[s_{i-1}+1,s_{i}]$. It follows by a straightforward induction that every curve $\Pi'$ with less that $i$ points will be at distance greater than $r$ from $C[1,s_{i-1}+1]$. The lemma now follows.
\end{proof}

\paragraph{Acknowledgment.}
We wish to thank Boris Aronov for helpful discussions on the problems studied in this paper.

\addcontentsline{toc}{section}{References}
\bibliographystyle{alphaurlinit}
\bibliography{refs}

\newpage

\appendix

\addcontentsline{toc}{section}{Appendix}
\addtocontents{toc}{\protect\setcounter{tocdepth}{0}}

\section{Remark on dimension reduction}\label{appendix:DimensionReduction}
In general, when the dimension $d$ is large, i.e. $d\gg\log(nm)$, one can use dimension reduction (using the celebrated Johnson-Lindenstrauss lemma~\cite{JL84}) in order to achieve a better running time, at the cost of inserting randomness in the prepossessing and query procedure.
However, such an approach can work only against an oblivious adversary, as it will necessarily fail for some curves.
Recently Narayanan and Nelson \cite{NN18} (improving \cite{EFN17,MMMR18}) proved a terminal version of the JL-lemma. Given a set $K$ of $k$ points in $\mathbb{R}^d$ and $\eps\in(0,1)$, there is a dimension reduction function $f:\mathbb{R}^d\rightarrow\mathbb{R}^{O(\frac{\log k}{\eps^2})}$ such that for every $x\in K$ and $y\in \mathbb{R}^d$ it holds that $\|x-y\|_2\le\|f(x)-f(y)\|_2\le(1+\eps)\cdot\|x-y\|_2$.

This version of dimension reduction can be used such that the query remains deterministic and always succeeds. The idea is to take all the $nm$ points from all the input curves to be the terminals, and let $f$ be the terminal dimension reduction. We transform each input curve $P = (p_1,\dotsc,p_m)$ into $f(P) = (f(p_1),\dotsc,f(p_m))$, a curve in $\mathbb{R}^{O(\frac{\log nm}{\eps^2})}$. Given a query $Q=(q_1,\dotsc,q_{m})$ we transform it to $f(Q) = (f(q_1),\dotsc,f(q_{m}))$. Since the pairwise distances between every query point to all input points are preserved, so is the distance between the curves. Specifically, the $\dpt$ distance w.r.t. any alignment $\tau$ is preserved up to a $1+\eps$ factor, and therefore we can reliably use the answer received using the transformed curves.

\section{A deterministic construction using a prefix tree}\label{appendix:deterministic-construction}
When implementing the dictionary $\D$ as a hash table, the construction of the data structure is randomized and thus in the worst case we might get higher prepeocessing time. To avoid this, we can implement $\D$ as a prefix tree.

\subsection{Discrete \frechet\ distance}
In this section we describe the implementation of $\D$ as a prefix tree in the case of ANNC under DFD.

We can construct a prefix tree $\T$ for the curves in $\I$, where any path in $\T$ from the root to a leaf corresponds to a curve that is stored in it.
For each $1\le i\le n$ and curve $\overline{Q}\in \I_i$, if $\overline{Q}\notin \T$, insert $\overline{Q}$ into $\T$, and set $C(\overline{Q})\leftarrow C_i$. 

Each node $v\in\T$ corresponds to a grid point from $\G$.
Denote the set of $v$'s children by $N(v)$. We store with $v$ a multilevel search tree on $N(v)$, with a level for each coordinate. 
The points in $\G$ are the grid points contained in $nm$ balls of radius $(1+\eps)r$. Thus when projecting these points to a single dimension, the number of 1-dimensional points is at most $nm\cdot\frac{\sqrt{d}(1+\eps)2r}{\eps r}=O(\frac{nm\sqrt{d}}{\eps})$. So in each level of the search tree on $N(v)$ we have $O(\frac{nm\sqrt{d}}{\eps})$ 1-dimensional points, so the query time is $O(d\log(\frac{nmd}{\eps}))$.

Inserting a curve of length $m$ to the tree $\T$ takes $O(md\log(\frac{nmd}{\eps}))$ time.
Since $\T$ is a compact representation of $|\I|=n\cdot O(\frac{1}{\eps})^{dm}$ curves of length $m$, the number of nodes in $\T$ is $m\cdot |\I|=nm\cdot O(\frac{1}{\eps})^{dm}$. Each node $v\in\T$ contains a search tree for its children of size $O(d\cdot|N(v)|)$, and $\sum_{v\in\T}|N(v)|=nm\cdot O(\frac{1}{\eps})^{dm}$ so the total space complexity is $O(nmd)\cdot O(\frac{1}{\eps})^{md}=n\cdot O(\frac{1}{\eps})^{md}$. Constructing $\T$ takes $O(|\I|\cdot md\log(\frac{nmd}{\eps}))=n\log(\frac{nmd}{\eps})\cdot O(\frac{1}{\eps})^{md}$ time.

\begin{theorem}
	There exists a data structure for the $(1+\eps,r)$-ANNC under DFD, with $n\cdot O(\frac{1}{\eps})^{dm}$ space, $n\cdot \log(\frac{n}{\eps})\cdot O(\frac{1}{\eps})^{md}$ preprocessing time, and $O(md\log(\frac{nmd}{\eps}))$ query time.
\end{theorem}

Similarly, for the asymmetric case we obtain the following theorem.
\begin{theorem}\label{thm:DFDassym-deterministic}
	There exists a data structure for the asymmetric $(1+\eps,r)$-ANNC under DFD, with $n\cdot O(\frac{1}{\eps})^{dk}$ space, $nm\log(\frac{n}{\eps})\cdot \left(O(d\log m)+O(\frac{1}{\eps})^{kd}\right)$ preprocessing time, and $O(kd\log(\frac{nkd}{\eps}))$ query time.
\end{theorem}

\subsection{$\ell_{p,2}$-distance}
For the case of ANNC under $\ell_{p,2}$-distance, the total number of curves stored in the tree $\T$ is roughly the same as in the case of DFD. We only need to show that for a given node $v$ of the tree $\T$, the upper bound on the size and query time of the search tree associated with it are similar.

The grid points corresponding to the nodes in $N(v)$ are from $n$ sets of $m$ balls with radius $(1+\eps)$.
When projecting the grid points in one of the balls to a single dimension, the number of 1-dimensional points is at most $\frac{m^{1/p}\sqrt{d}}{\eps}\cdot (1+\eps)$, so the total number of projected points is at most $\frac{nm^{1+\frac1p}\sqrt{d}}{\eps}\cdot (1+\eps)$.

Thus in each level of the search tree of $v$ we have $O(\frac{nm^2\sqrt{d}}{\eps})$ 1-dimensional points, so the query time is $O(d\log(\frac{nmd}{\eps}))$, and inserting a curve of length $m$ into the tree $\T$ takes $O(md\log(\frac{nmd}{\eps}))$ time. Note that the size of the search tree of $v$ remains $O(d\cdot|N(v)|)$.

We conclude that the total space complexity is $O(\frac{nm^2\sqrt{d}}{\eps})\cdot O(\frac{1}{\eps})^{m(d+1)}=n\cdot O(\frac{1}{\eps})^{m(d+1)}$, constructing $\T$ takes $O(|\I|\cdot md\log(nmd/\eps))=n\log(\frac{n}{\eps})\cdot O(\frac{1}{\eps})^{m(d+1)}$ time, and the total query time is $O(md\log(\frac{nmd}{\eps}))$.

\begin{theorem}
	There exists a data structure for the $(1+\eps,r)$-ANNC under $\ell_{p,2}$-distance, with $n\cdot O(\frac{1}{\eps})^{m(d+1)}$ space, $n\cdot \log(\frac{n}{\eps})\cdot O(\frac{1}{\eps})^{m(d+1)}$ preprocessing time, and $O(md\log(\frac{nmd}{\eps}))$ query time.
\end{theorem}

\section{Dealing with query curves and input curves of varying size}\label{appendix:DifferentSizes}

Notice that if an input curve $C_i$ has length $t<m$, then the size of the set of candidates $\I_i$ (and $\I'_i$ in the asymmetric case) can only decrease.

In addition, our assumption that all query curves are of length exactly $k$ can be easily removed by constructing $k$ data structures $\D_1,\dots,\D_k$, where $\D_i$ is our data structure constructed for query curves of length $i$ (instead of $k$), for $1 \le i \le k$. Clearly, the query time does not change. The storage space is multiplied by $k$, so for the case of DFD we have storage space $nk\cdot O(\frac{1}{\eps})^{kd}$, but $k<2^{kd}$, so the storage space remains $n\cdot O(\frac{1}{\eps})^{kd}$. Similarly, for the case of $\ell_{p,2}$-distance we obtain storage space of $n\cdot O(\frac{1}{\eps})^{k(d+1)}\cdot\left(\frac{m}{k}\right)^{kd/p}$.

\section{One-way alignments}\label{appendix:one-way}
\begin{claim}\label{clm:oneWayTriangleEnq}
	Let $A,B,C$ be three curves, and let $\tau_1$, $\tau_2$ be two one-way alignments such that $\tau_1$ matches $C$ to $A$ and $\tau_2$ matches $C$ to $B$.
	Then $\dpt(A,B)\le \sigma_{p,2}(\tau_1(C,A))+\sigma_{p,2}(\tau_2(C,B))$.
\end{claim}
\begin{proof}
	Denote by $k_A,k_B,k_C$ the lengths of the curves $A,B,C$ respectively.
	Consider the following algorithm that constructs an alignment $\tau$. 
	For every $1\le x\le k_C$, denote by $i_x,j_x$ the unique indexes such that $(x,i_x)\in \tau_1$ and $(x,j_x)\in \tau_2$. Add the pair $(i_x,j_x)$ to $\tau$ if it is not already there.
	
	First, we need to show that $\tau=\langle(i_1,j_1),\dots,(i_t,j_t)\rangle$ is a valid alignment.
	Clearly, $(i_1,j_1)=(1,1)$ because $(1,1)\in \tau_1$ and $(1,1)\in \tau_2$.
	Similarly, $(i_t,j_t)=(k_A,k_B)$ because $(k_C,k_A)\in \tau_1$ and $(k_C,k_B)\in \tau_2$.
	
	For any $1\le s<t$, consider the two consecutive pairs $(i_s,j_s),(i_{s+1},j_{s+1})\in \tau$. Let $x_1$ be an index such that $(x_1,i_s)\in \tau_1$ and $(x_1,j_s)\in \tau_2$, and $x_2$ an index such that $(x_2,i_{s+1})\in \tau_1$ and $(x_2,j_{s+1})\in \tau_2$. Since $\tau_1,\tau_2$ are one-way alignments, we have $x_1\neq x_2$. Moreover, since the algorithm added $(i_s,j_s)$ to $\tau$ before $(i_{s+1},j_{s+1})$, we have $x_1<x_2$. This implies that $i_{s+1}\ge i_s$ and $j_{s+1}\ge j_s$. Assume by contradiction that $i_{s+1} > i_s+1$, and let $x$ be the index such that $(x,i_s+1)\in\tau_1$, then $x_1<x<x_2$ and thus the algorithm adds a pair $(i_s+1,j)$ for some index $j$ after $(i_s,j_s)$ and before $(i_{s+1},j_{s+1})$, a contradiction. So we have $i_s\le i_{s+1} \le i_s+1$, and by symmetric arguments, $j_s\le j_{s+1} \le j_s+1$, and therefore $\tau$ is valid.
	
	Using the triangle inequality for the $\ell_p$ norm, we get that 
	\begin{align*}
	\dpt(A,B)\le\sigma_{p,2}(\tau(A,B)) & =\Big(\sum_{(i,j)\in\tau}\|a_{i}-b_{j}\|_{2}^{p}\Big)^{1/p}\\
	& \le\Big(\sum_{x=1}^{k_{C}}\|a_{i_{x}}-b_{j_{x}}\|_{2}^{p}\Big)^{1/p}\\
	& \le\Big(\sum_{x=1}^{k_{C}}\|a_{i_{x}}-c_{x}\|_{2}^{p}\Big)^{1/p}+\Big(\sum_{x=1}^{k_{C}}\|c_{x}-b_{j_{x}}\|_{2}^{p}\Big)^{1/p}\\
	& =\sigma_{p,2}(\tau_{1}(C,A))+\sigma_{p,2}(\tau_{2}(C,B))~.
	\end{align*}
\end{proof}

\end{document}